\documentclass[twoside]{article}
\usepackage{amsmath}
\usepackage{amsfonts}
\usepackage{amsthm}
\usepackage{graphicx}
\usepackage{subfigure}
\usepackage{natbib}

\newtheorem{theorem}{Theorem}[section]
\newtheorem{lemma}[theorem]{Lemma}
\newtheorem{prop}[theorem]{Proposition}
\newtheorem{definition}[theorem]{Definition}
\newtheorem{remark}[theorem]{Remark}
\newcommand{\Keywords}[1]{\par\noindent
{\small{\bf Keywords\/}: #1}}

\title{\sc optimal execution strategy in the presence of permanent price impact and fixed transaction cost}

\author{Mauricio Junca\thanks{Department of Industrial Engineering and Operations Research, University of California, Berkeley CA 94720-1777. Email address: {\tt mjunca@berkeley.edu}}}

\begin{document}

\maketitle

\begin{abstract}
We study a single risky financial asset model subject to price impact and transaction cost over an infinite horizon. An investor needs to execute a long position in the asset affecting the price of the asset and possibly incurring in fixed transaction cost. The objective is to maximize the discounted revenue obtained by this transaction. This problem is formulated first as an impulse control problem and we characterize the value function using the viscosity solutions framework. We also analyze the case where there is no transaction cost and how this formulation relates with a singular control problem. A viscosity solution characterization is provided in this case as well. We also establish a connection between both formulations with zero fixed transaction cost. Numerical examples with different types of price impact conclude the discussion.
\newline
\end{abstract}

\pagestyle{myheadings}
\markboth{M. JUNCA}{OPTIMAL EXECUTION WITH PRICE IMPACT}

\section{Introduction}

An important problem for stock traders is to unwind large block orders of shares. Market microstructure literature has shown (e.g. \citep{Chan,Holthausen}), both theoretically and empirically, that large trades move the price of risky securities either for informational or liquidity reasons. Several papers addressed this issue and formulated a hedging and arbitrage pricing theory for large investors under competitive markets. For example, in \citep{Cvitanic} a forward-backward SDE is defined, with the price process being the forward component and the wealth process of the investor's portfolio being the backward component. In both cases, the drift and volatility coefficients depend upon the price of the stocks, the wealth of the portfolio and the portfolio itself. \citep{Frey} describes the discounted stock price using a reaction function that depends on the position of the large trader. In \citep{Bank,CetinLiq} the authors, independently, described the price impact by assuming a given family of continuous semi-martingales indexed by the number of shares held (\citep{Bank}) and by the number of shares traded (\citep{CetinLiq}).

The optimal execution problem has been studied in \citep{Bertsimas, Almgren} in a discrete-time framework. In both cases the dynamics of the price processes are arithmetic random walks affected by the trading strategy. In \citep{Bertsimas} the impact is proportional to the amount of shares traded. In \citep{Almgren} the change in the price is twofold, a temporary impact caused by temporary imbalances in supply/demand dynamics and a permanent impact in the equilibrium or unperturbed price process due to the trading itself. Also, this work takes into account the variance of the strategy with a mean-variance optimization procedure. Later on, nonlinear price impact functions were introduced in \citep{Almgren2}. These ideas were adopted by more recent works under a continuous time framework. \citep{Schied3} propose the problem within a regular control setting. The authors consider expected-utility maximization for CARA utility functions, that is, for exponential utility functions. The dynamics of the price and the market impact function are fairly general. \citep{Schied2} is the only reference that considers an infinite horizon model based on the original model in \citep{Almgren}. Finally, \citep{He} consider only permanent price impact but they allow continuous and discrete trading (singular control setting) with a geometric Brownian motion as price process.

On the other hand, it is also well established that transaction costs in asset markets are an important factor in determining the trading behavior of market participants. Typically, two types of transaction costs are considered in the context of optimal consumption and portfolio optimization: proportional transaction costs (\citep{Davis,oksulem}) using singular type controls and fixed transaction costs (\citep{Korn,oksulem}) using impulse type controls. The market impact effect can be significantly reduced by splitting the order into smaller orders but this will increase the transaction cost effect. Thus, the question is to find optimal times and allocations for each individual placement such that the expected revenue after trading is maximized. \citep{Ly} include both permanent market price impact and fixed transaction cost and assume that the unperturbed price process is a geometric Brownian motion process. This reference only accepts discrete trading (impulse control setting) and uses the theory of (discontinuous) viscosity solutions to characterize the value function. Finally, \citep{ajay} propose a slightly different model which does not include any fixed transaction cost but includes an execution lag associated with size of the discrete trades. It is important to remark that all papers referenced above assume a terminal date at which the investor must liquidate her position.

In this paper we study an infinite horizon price impact model that includes fixed transaction cost under the setting of impulse control, similar to \citep{Ly}. We describe a general underlying price process and a general permanent market impact. With help of some classic results for optimal stopping problems and the discontinuous viscosity solutions theory for nonlinear partial differential equations, developed in references such as \citep{crisli,ishii1,ishii2,Fleming}, we obtain a fully characterization of the value function when the price process satisfies some technical condition. Most of the processes used in financial studies satisfy this condition. This characterization is not complete when the fixed transaction cost is zero. By analyzing the HJB equation obtained in the previous case, we formulate a singular control model to include this case. For this new formulation we are able to complete the characterization. We are able to show that both formulations agree in the value function even though the formulations are completely different, when we choose the appropriate market impact functions. Finally, we describe the value function and the optimal strategy explicitly for an important special case.

The structure of the paper is as follows: The general impulse control model, growth condition and boundary properties of the value function which are useful for the characterization of the function are exposed in Section 2.  This section characterizes the value function of the problem as a viscosity solution of the Hamilton-Jacobi-Bellman equation and shows uniqueness when the fixed transaction cost is strictly positive and the price process satisfies certain conditions. Section 3 proposes a singular control model to tackle the case when the transaction cost is zero. Here a viscosity solution characterization and uniqueness result are proved as well under the same conditions. We present a special case where the value function of the impulse control model coincides with the value function of the singular control model and obtain the value function explicitly for this case. Section 4 shows that, in fact, these two formulation produce the same value function even though they consider different types of control. Section 5 presents numerical results for different underlying price processes with both formulations. Finally, we state some conclusions and future work.

\section{Impulse control model}\label{model}

Let $(\Omega,\mathcal{F},(\mathcal{F}_t)_{0\leq t\leq \infty},\mathbb{P})$ be a probability space which satisfies the usual conditions and $B_t$ be a one-dimensional Brownian motion adapted to the filtration. We consider a continuous time process adapted to the filtration denoting the price of a risky asset $P_t$. The unperturbed price dynamics, when the investor makes no action, are given by:
\begin{align}
dP_s&=\mu(P_s)ds+\sigma(P_s)dB_s,
\end{align}
where $\mu$ and $\sigma$ satisfy regular conditions such that there is a unique strong solution of this SDE (i.e. Lipschitz continuity). We are mainly interested in price processes that are always non-negative, thus we assume that $P$ is absorbed as soon as it reaches 0 and that the initial price $p$ is non-negative.

Now, the number of shares in the asset held by the investor at time $t$ is denoted by $X_t$ and it is up to the investor to decide how to unwind the shares. Different models and formulations will define the admissible strategies for the investor. At the beginning the investor has $x\geq0$ number of shares and we only allow strategies such that $X_t\geq0$ for all $t\geq0$. Since the investor's interest is to execute the position, we don't allow to buy shares, that is $X_t$ is a non-increasing process.  Hence, we can see that $\mathbb{R}_+\times\mathbb{R}_{+}=\bar{\mathcal{O}}$ (with interior $\mathcal{O}$) is the state space of the problem. The goal of the investor is to maximize the expected discounted profit obtained by selling the shares. Given $y=(x,p)\in\bar{\mathcal{O}}$ we define $V(y)$, the value function, as such maximum (or supremum), taken over all admissible trading strategies such that $(X_{0-},P_{0-})=Y_{0-}=y$. We call $\beta>0$ the discount factor and $k\geq0$ the transaction cost. Note that we can always do nothing, in which case the expected revenue is 0. Therefore $V\geq0$ for all $y$.

In this formulation we assume that the investor can only trade discretely over the time horizon. This is modeled with the impulse control $\nu=(\tau_n,\zeta_n)_{1\leq n\leq M}$, where the random variable $M<\infty$ is the number of trades, $(\tau_n)$ are stopping times with respect to the filtration $(\mathcal{F}_t)$ such that $0\leq\tau_1\leq\cdots\leq\tau_n\leq\cdots\leq\tau_M\leq\infty$ that represent the times of the investor's trades, and $(\zeta_n)$ are real-valued $\mathcal{F}_{\tau_n}$-measurable random variables that represent the number of shares sold at the intervention times. Note that any control policy $\nu$ fully determines $M$. Given any strategy $\nu$, the dynamics of $X$ are given by
\begin{align}
X_s&=X_{\tau_n},\textrm{  for  }\tau_n\leq s<\tau_{n+1},\label{xdyn1}\\
X_{\tau_{n+1}}&=X_{\tau_n}-\zeta_{n+1}.\label{xdyn2}
\end{align}
We consider price impact functions such that the price goes down when the investor sells shares. Also, the greater the volume of the trade, the grater the impact in the price process. Thus, we let $\alpha(\zeta,p)$ be the post-trade price when the investor trades $\zeta$ shares of the asset at a pre-trade price of $p$. We assume that $\alpha$ is smooth, non-increasing in $\zeta$, and non-decreasing in $p$ such that $\alpha(0,p)=p$ for all $p$. These conditions imply that $\alpha(\zeta,p)\leq p$ for $\zeta\geq0$. Furthermore, we will also assume that for all $\zeta_1,\zeta_2,p\in\mathbb{R_+}$
\begin{equation}\label{suma}
\alpha(\zeta_1,\alpha(\zeta_2,p))=\alpha(\zeta_1+\zeta_2,p).
\end{equation}
This assumption says that the impact in the price of trading twice at the same moment in time is the same as trading the total number of shares once. This assumption will prevent any price manipulation from the investor. Two possible choices for $\alpha$ are:
$$\alpha_1(\zeta,p)=p-\lambda\zeta$$
$$\alpha_2(\zeta,p)=pe^{-\lambda\zeta}$$
where $\lambda>0$. A linear impact like $\alpha_1$ has the drawback that the post-trade price can be negative. Given a price impact $\alpha$ and an admissible strategy $\nu$, the price dynamics are given by:
\begin{align}
dP_s&=\mu(P_s)ds+\sigma(P_s)dB_s,\textrm{  for  }\tau_n\leq s<\tau_{n+1},\label{pdyn1}\\
P_{\tau_n}&=\alpha(\zeta_n,P_{\tau_{n-}}).\label{pdyn2}
\end{align}
When $\tau_n=\tau_{n+1}$, then we apply the impact twice, therefore 
$$P_{\tau_n}=P_{\tau_{n+1}}=\alpha(\zeta_{n+1},\alpha(\zeta_n,P_{\tau_{n-}})).$$
If more that two actions are taken at the same time, we apply the impact accordingly. Now, given $y=(x,p)\in\bar{\mathcal{O}}$ the value function $V$ has the form:
\begin{equation}\label{valuefunctioninf}
V(y)=\sup_{\nu}\mathbb{E}\left[\sum\limits_{n=1}^{M}e^{-\beta\tau_n}(\zeta_nP_{\tau_n}-k)\right].
\end{equation}
As usual, we assume that $e^{-\beta\tau}=0$ on $\{\tau=\infty\}$.

\subsection{Hamilton-Jacobi-Bellman equation}

In order to characterize the value function we will use the dynamic programming approach. This principle has been proved for several frameworks and types of control. Some of the references that prove it in a fairly general context are \citep{Ishikawa,Ma}. We have that the following Dynamic Programming Principle (DPP) holds: For all $y=(x,p)\in\mathcal{O}$ we have
\begin{equation}\label{dppinf}
V(y)=\sup_{\nu}\mathbb{E}\left[\sum_{\tau_n\leq\tau}e^{-\beta\tau_n}(\zeta_nP_{\tau_n}-k)+e^{-\beta\tau}V(Y_{\tau})\right],
\end{equation}
where $\tau$ is any stopping time. Let's define the impulse transaction function as
$$\Gamma(y,\zeta)=(x-\zeta,\alpha(\zeta,p))$$
for all $y\in\bar{\mathcal{O}}$ and $\zeta\in\mathbb{R}$. This corresponds to the change in the state variables when a trade of $\zeta$ shares has taken place. We define the intervention operator as
$$\mathcal{M}\varphi(y)=\sup_{0\leq\zeta\leq x}\varphi(\Gamma(y,\zeta))+\zeta\alpha(\zeta,p)-k,$$
for any measurable function $\varphi$.  Also, let's define the infinitesimal generator operator associated with the price process when no trading is done, that is
$$A\varphi=\mu(p)\frac{\partial\varphi}{\partial p}+\frac{1}{2}\sigma(p)^2\frac{\partial^2\varphi}{\partial p^2},$$
for any  function $\varphi\in C^2(\mathcal{O})$. The HJB equation that follows from the DPP is then (\citep{oeksendal05acjd})
\begin{equation}\label{hjbinf}
\min\left\{\beta\varphi-A\varphi,\varphi-\mathcal{M}\varphi\right\}=0\textrm{ in } \mathcal{O}.
\end{equation}
We call the continuation region to 
$$\mathcal{C}=\{y\in\mathcal{O}:\mathcal{M}\varphi-\varphi<0\}$$
and the trade region to
$$\mathcal{T}=\{y\in\mathcal{O}:\mathcal{M}\varphi-\varphi=0\}.$$

\subsection{Growth Condition}\label{grosec}

We will define a particular optimal stopping problem and use some of the results in \citep{dayanik} to establish an upper bound on the value function $V$ and therefore a growth condition. Consider the case where there is no price impact, that is, $\alpha(\zeta,p)=p$ for all $\zeta\geq0$.  We define
\begin{equation}\label{noimpact}
V_{NI}(y)=\sup_{\nu}\mathbb{E}\left[\sum\limits_{n=1}^{M}e^{-\beta\tau_n}(\zeta_nP_{\tau_n}-k)\right],
\end{equation}
where $P_s$ follows the unperturbed price process. It is clear that $V\leq V_{NI}$. When there is no price impact, the investor would need to trade only one time.
\begin{prop}\label{growprop}
For all $y\in\mathcal{O}$
\begin{equation}\label{optstoprop}
V_{NI}(x,p)= U(x,p):=\sup_{\tau}\mathbb{E}[e^{-\beta\tau}(xP_{\tau}-k)^+]
\end{equation}
where the supremum is taken over all stopping times with respect to the filtration $(\mathcal{F}_s)$.
\end{prop}
\begin{proof}
Since $(\tau,x)$ is an admissible strategy for any stopping time $\tau$, then $U\leq V_{NI}$. Now, let $\Upsilon_n$ the set of admissible strategies with at most $n$ interventions. The proof will continue by induction in $n$ to show that for all $n$
\begin{equation}\label{induc}
\sup\limits_{\nu\in\Upsilon_n}\mathbb{E}\left[\sum\limits_{i=1}^{n}e^{-\beta\tau_i}(\zeta_iP_{\tau_i}-k)\right]\leq U(y).
\end{equation}
Clearly \eqref{induc} is true for $n=1$. Let $\nu\in\Upsilon_n$. Note that $xp-k\leq U(x,p)$, therefore, conditioning on $\mathcal{F}_{\tau_1}$ we have
\begin{align*}
\mathbb{E}\left[\sum\limits_{i=1}^{n}e^{-\beta\tau_i}(\zeta_iP_{\tau_i}-k)\right]=&\mathbb{E}\left[\mathbb{E}\left[\left.[e^{-\beta\tau_1}(\zeta_1P_{\tau_1}-k)\right|\mathcal{F}_{\tau_1}\right]\right]+\\
&\mathbb{E}\left[\mathbb{E}\left[\left.e^{-\beta\tau_1}\sum\limits_{i=2}^{n}e^{-\beta(\tau_i-\tau_1)}(\zeta_iP_{\tau_i}-k)\right|\mathcal{F}_{\tau_1}\right]\right]\\
\leq&\mathbb{E}\left[\mathbb{E}\left[\left.e^{-\beta\tau_1}U(\zeta_1,P_{\tau_1})\right|\mathcal{F}_{\tau_1}\right]\right]+\\
&\mathbb{E}\left[\mathbb{E}\left[\left.e^{-\beta\tau_1}U(x-\zeta_1,P_{\tau_1})\right|\mathcal{F}_{\tau_1}\right]\right]\\
\leq&\mathbb{E}\left[e^{-\beta\tau_1}U(x,P_{\tau_1})\right]\\
\leq&U(x,p),
\end{align*}
where the last inequality follows from the fact that the process $e^{-\beta s}U(x,P_s)$ is a supermartingale (\citep{oksenreik}). This proves \eqref{induc}. By Lemma 7.1 in \citep{oeksendal05acjd}, the left hand side of \eqref{induc} converges to $V_{NI}$ as $n\rightarrow\infty$ and the proof is complete.
\end{proof}
From the previous result we have the bound
\begin{equation}\label{optstop}
0\leq V(x,p)\leq U(x,p)=\sup_{\tau}\mathbb{E}[e^{-\beta\tau}(xP_{\tau}-k)],
\end{equation}
where the supremum is taken over all stopping times with respect to the filtration $(\mathcal{F}_t)$. Following section 5 in \citep{dayanik}, let $\psi$ and $\phi$ be the unique, up to multiplication by a positive constant, strictly increasing and strictly decreasing (respectively) solutions of the ordinary differential equation $Au=\beta u$ and such that $0\leq\psi(0+)$ and $\psi(p)\rightarrow\infty$ as $p\rightarrow\infty$. For any $x\geq0$, let
\begin{equation}\label{ll}
\ell_x=\lim\limits_{p\rightarrow\infty}\frac{(xp-k)^+}{\psi(p)}.
\end{equation}
Then $U$ is finite in $\mathcal{O}$ if and only if $\ell_x$ is finite for all $x\geq0$. Furthermore, when $U$ is finite we also have that for some $C>0$
\begin{equation}\label{boundinf}
U(x,p)\leq Cx\psi(p)
\end{equation}
and
\begin{equation}
\lim\limits_{p\rightarrow\infty}\frac{U(x,p)}{\psi(p)}=\ell_x.
\end{equation}
We will assume that $U$ is finite.

\subsection{Boundary Condition}\label{bousec}

Since the investor is not allowed to purchase shares of the asset we have that $V(0,p)=0$ for all $p\geq0$. Also, the price process gets absorbed at 0, therefore $V=0$ on $\partial\mathcal{O}$. If we assume that $U$ is finite then by \eqref{boundinf} we have that $V(x,p)\rightarrow0$ as $x\rightarrow0$ for all $p\geq0$, that is, $V$ is continuous on $\{x=0\}$. Now we distinguish two cases:\\
\begin{enumerate}
\item 0 is an absorbing boundary for the price process $P$. This means that for any $p>0$, $\mathbb{P}(P_t=0 \textrm{ for some }t>0|P_0=p)>0$. A simple example is the arithmetic Brownian motion. Since the process is stopped at 0, we must have that for all $x\geq0$
\begin{equation*}\label{absbound}
U(x,0)=0.
\end{equation*}
Also, \citep{dayanik} shows that in this case $U$ is continuous at $\{p=0\}$ whenever U is finite. Therefore the boundary conditions for the value function $V$ are
\begin{equation}\label{boundaryinfab}
V=0\textrm{ on } \partial\mathcal{O} \textrm{ and }\lim_{y'\rightarrow y}V(y')=0 \textrm{ for all }y\in\partial\mathcal{O}.
\end{equation}
\item 0 is a natural boundary for the price process $P$. This means that for any $p>0$, $\mathbb{P}(P_t=0 \textrm{ for some }t>0|P_0=p)=0$. For example the geometric Brownian motion. In this case we can have different situations in $V(x,p)$ as $p$ goes to 0 depending on the price process. In particular, we can have the situation where $V$ is discontinuous on the set $\{p=0\}$.
\end{enumerate}

\subsection{Viscosity solution}

We now are going to prove that the value function is a viscosity solution of the HJB equation \eqref{hjbinf} and find the appropriate conditions that make this value function unique. The appropriate notion of solution of the HJB equation \eqref{hjbinf} is the notion of discontinuous viscosity solution since we cannot know a priori if the value function is continuous in $\mathcal{O}$. We must first state some definitions.
\begin{definition}
Let $W$ be an extended real-valued function on some open set $\mathcal{D}\subset\mathbb{R}^n$.
\begin{enumerate}
\item[(i)] The \emph{upper semi-continuous envelope} of $W$ is
$$W^{*}(x)=\lim_{r\downarrow0}\sup_{{|x'-x|\leq r\atop x'\in\mathcal{D}}}W(x')\textrm{,   }\forall x\in\mathcal{D}.$$
\item[(ii)] The \emph{lower semi-continuous envelope} of $W$ is
$$W_{*}(x)=\lim_{r\downarrow0}\inf_{{|x'-x|\leq r\atop x'\in\mathcal{D}}}W(x')\textrm{,   }\forall x\in\mathcal{D}.$$
\end{enumerate}
\end{definition}

Note that $W^*$ is the smallest upper semi-continuous function which is greater than or equal to $W$, and similarly for $W_*$. Now we define discontinuous viscosity solutions.

\begin{definition}\label{defviscosity}
Given an equation of the form
\begin{equation}\label{hjbgen}
\min\left\{F(x,\varphi(x),D\varphi(x),D^2\varphi(x)),\varphi-\mathcal{M}\varphi\right\}=0\textrm{ in } \mathcal{D},
\end{equation}
a locally bounded function $W$ on $\mathcal{D}$ is a:
\begin{enumerate}
\item[(i)] \emph{Viscosity subsolution} of \eqref{hjbgen} in $\mathcal{D}$ if for each $\varphi\in C^{2}(\bar{\mathcal{D}})$,
$$\min\left\{F(x_0,W(x_0),D\varphi(x_0),D^2\varphi(x_0)),W^*(x_0)-\mathcal{M}W^*(x_0)\right\}\leq0$$ 
at every $x_0\in\mathcal{D}$ which is a maximizer of $W^*-\varphi$ on $\bar{\mathcal{D}}$ with $W^*(x_0)=\varphi(x_0)$.

\item[(ii)] \emph{Viscosity supersolution} of \eqref{hjbgen} in $\mathcal{D}$ if for each $\varphi\in C^{2}(\bar{\mathcal{D}})$,
$$\min\left\{F(x_0,W(x_0),D\varphi(x_0),D^2\varphi(x_0)),W_*(x_0)-\mathcal{M}W_*(x_0)\right\}\geq0$$ 
at every $x_0\in\mathcal{D}$ which is a minimizer of $W_*-\varphi$ on $\bar{\mathcal{D}}$ with $W_*(x_0)=\varphi(x_0)$.

\item[(iii)] \emph{Viscosity solution} of \eqref{hjbgen} in $\mathcal{D}$ if it is both a viscosity subsolution and a viscosity supersolution of \eqref{hjbgen} in $\mathcal{D}$.
\end{enumerate}
\end{definition}

We are now ready for the following theorem:

\begin{theorem}\label{viscosity}
The value function $V$ defined by \eqref{valuefunctioninf} is a viscosity solution of \eqref{hjbinf} in $\mathcal{O}$.
\end{theorem}
\begin{proof} By the bounds given in the section \ref{grosec}, it is clear that $V$ is locally bounded. Now we show the viscosity solution property.

Subsolution property: Let $y_0\in\mathcal{O}$ and $\varphi\in C^{2}(\mathcal{O})$ such that $y_0$ is a maximizer of $V^*-\varphi$ on $\mathcal{O}$ with $V^*(y_0)=\varphi(y_0)$. Now suppose that there exists $\theta>0$ and $\delta>0$ such that
\begin{equation}\label{contr2inf}
-\beta\varphi(y)+A\varphi(y)\leq-\theta
\end{equation}
for all $y\in\mathcal{O}$ such that $|y-y_0|<\delta$. Let $(y_n)$ be a sequence in $\mathcal{O}$ such that $y_n\rightarrow y_0$ and
$$\lim_{n\rightarrow\infty}V(y_n)=V^*(y_0).$$ 
By the dynamic programming principle \eqref{dppinf}, for all $n\geq1$ there exist an admissible control $\nu_n=(\tau^n_m,\zeta^n_m)_m$ such that for any stopping time $\tau$ we have that
\begin{equation}\label{approxinf}
V(y_n)\leq\mathbb{E}\left[\sum_{\tau^n_m\leq\tau}e^{-\beta\tau^n_m}(\zeta^n_mP^n_{\tau^n_m}-k)+e^{-\beta\tau}V(Y^n_{\tau})\right]+\frac{1}{n},
\end{equation}
where $Y_s^n$ is the process controlled by $\nu_n$ for $s\geq 0$. Now consider the stopping time
$$T_n=\inf\{s\geq 0:|Y_s^n-y_0|\geq\delta\}\wedge\tau_1^n,$$
where $\tau_1^n$ is the first intervation time of the impulse control $\nu_n$. 
By \eqref{approxinf} we have that
\begin{align}
V(y_n)&\leq\mathbb{E}\left[e^{-\beta T_n}V(Y^n_{T_n})1_{\{T_n<\tau_1^n\}}\right]+\mathbb{E}\left[e^{-\beta T_n}\left(\zeta_1^nP^n_{\tau_1^n}-k+V(Y^n_{\tau_1^n})\right)1_{\{T_n=\tau_1^n\}}\right]+\frac{1}{n}\nonumber\\
&\leq\mathbb{E}\left[e^{-\beta T_n}V(Y^n_{T_n-})1_{\{T_n<\tau_1^n\}}\right]+\mathbb{E}[e^{-\beta T_n}\mathcal{M}V(Y^n_{\tau_1^n-})1_{\{T_n=\tau_1^n\}}]+\frac{1}{n}\label{aux2inf}\\
&\leq\mathbb{E}\left[e^{-\beta T_n}V(Y^n_{T_n-})\right]+\frac{1}{n}\label{aux3inf}
\end{align}
Now, by Dynkin's formula and \eqref{contr2inf} we have 
\begin{align*}
\mathbb{E}[e^{-\beta T_n}\varphi(Y^n_{T_n-})]&=\varphi(y_n)+\mathbb{E}\left[\int_{0}^{T_n}e^{-\beta s}\left (-\beta \varphi(Y_s^n)+A\varphi(Y_s^n)\right)ds\right]\\
&\leq\varphi(y_n)-\frac{\theta}{\beta}(1-\mathbb{E}[e^{-\beta T_n}]).
\end{align*}
Since $V\leq V^*\leq \varphi$ and $T_n\leq\tau_1^n$, by \eqref{aux3inf}
$$V(y_n)\leq\varphi(y_n)-\frac{\theta}{\beta}(1-\mathbb{E}[e^{-\beta T_n}])+\frac{1}{n},$$
for all $n$. Letting $n$ go to infinity we have that 
$$\lim_{n\rightarrow\infty}\mathbb{E}[e^{-\beta T_n}]=1,$$ 
which implies that
$$\lim_{n\rightarrow\infty}\mathbb{P}[\tau^n_1=0]=1.$$
Combining the above with \eqref{aux2inf} when we let $n\rightarrow\infty$ we get
$$V^*(y_0)\leq\sup_{|y'-y_0|<\delta}\mathcal{M}V(y').$$
Since this is true for all $\delta$ small enough, then sending $\delta$ to 0 we have
$$V^*(y_0)\leq(\mathcal{M}V)^*(y_0).$$
If we show that $(\mathcal{M}V)^*\leq\mathcal{M}V^*$, then we would have proved that if $-\beta\varphi(y_0)+A\varphi(y_0)<0$, then $\mathcal{M}V^*(y_0)-V^*(y_0)\geq0$ and therefore
$$\min\left\{\beta\varphi(y_0)-A\varphi(y_0),V^*(y_0)-\mathcal{M}V^*(y_0)\right\}\leq0.$$
Appendix \ref{lemita} contains the proof of this last fact.

Supersolution property: Let $y_0\in\mathcal{O}$ and $\varphi\in C^{2}(\mathcal{O})$ such that $y_0$ is a minimizer of $V_*-\varphi$ on $\mathcal{O}$ with $V_*(y_0)=\varphi(y_0)$. By definition of $V$ and $\mathcal{M}V$ we have that $\mathcal{M}V\leq V$ on $\mathcal{O}$ and therefore $(\mathcal{M}V)_*\leq V_*$. Let $(y_n)$ be a sequence in $\mathcal{O}$ such that $y_n\rightarrow y_0$ and
$$\lim_{n\rightarrow\infty}V(y_n)=V_*(y_0).$$
Now, since $V_*\leq V$ is lower semi-continuous and $\Gamma$ is continuous we have
\begin{align*}
\mathcal{M}V_*(y_0)&=\sup_{0\leq\zeta\leq x_0}V_*(\Gamma(y_0,\zeta))+\zeta\alpha(\zeta,p_0)-k\\
&\leq\sup_{0\leq\zeta\leq x_0}\liminf_{n\rightarrow\infty}V(\Gamma(y_n,\zeta))+\zeta\alpha(\zeta,p_n)-k\\
&\leq\liminf_{n\rightarrow\infty}\sup_{0\leq\zeta\leq x_n}V(\Gamma(y_n,\zeta))+\zeta\alpha(\zeta,p_n)-k\\
&\leq\lim_{n\rightarrow\infty}\mathcal{M}V(y_n)\\
&=(\mathcal{M}V)_*(y_0).
\end{align*}
Hence $\mathcal{M}V_*(y_0)\leq(\mathcal{M}V)_*(y_0)\leq V_*(y_0)$. Now suppose that there exists $\theta>0$ and $\delta>0$ such that
\begin{equation}\label{contr1inf}
\beta\varphi(y)-A\varphi(y)\leq-\theta
\end{equation}
for all $y\in\mathcal{O}$ such that $|y-y_0|<\delta$. Fix $n$ large enough such that $|y_n-y_0|<\delta$ and consider the process $Y_s^n$ for $s\geq 0$ with no intervention such that $Y^n_{0}=y_n$. Let 
$$T_n=\inf\{s\geq 0:|Y_s^n-y_0|\geq\delta\}.$$
Now, by Dynkin's formula and \eqref{contr1inf} we have 
\begin{align*}
\mathbb{E}[e^{-\beta T_n}\varphi(Y^n_{T_n})]&=\varphi(y_n)+\mathbb{E}\left[\int_{0}^{T_n}e^{-\beta s}\left (-\beta \varphi(Y_s^n)+A\varphi(Y_s^n)\right)ds\right]\\
&\geq\varphi(y_n)+\frac{\theta}{\beta}(1-\mathbb{E}[e^{-\beta T_n}]).
\end{align*}
On the other hand, $\varphi\leq V_*\leq V$ and using the dynamic programming principle \eqref{dppinf} we have
$$\mathbb{E}[e^{-\beta T_n}\varphi(Y^n_{T_n})]\leq\mathbb{E}[e^{-\beta T_n}V(Y^n_{T_n})]\leq V(y_n).$$
Notice that $\eta:=\lim\limits_{n\rightarrow\infty}\mathbb{E}[e^{-\beta T_n}]<1$ since $T_n>0$ a.s by a.s continuity of the processes $Y_s^n$, then by the above two inequalities and taking $n\rightarrow\infty$, we have that
$$V_*(y_0)\geq\varphi(y_0)+\frac{\theta}{\beta}(1-\eta)>\varphi(y_0)$$
contradicting the fact that $V_*(y_0)=\varphi(y_0)$. This establishes the supersolution property.
\end{proof}

\subsection{Uniqueness}

Let $\psi$ be defined as before and let's assume that the function $U$ defined in \eqref{optstop} is finite. Also assume that the transaction cost $k>0$. Then, we want to prove that $V$ is the unique viscosity solution of the equation \eqref{hjbinf} that is bounded by $U$. We will need an additional assumption about the function $\psi$: For all $x\geq0$
\begin{equation}\label{psi}
\lim\limits_{p\rightarrow\infty}\frac{U(x,p)}{\psi(p)}=\ell_x=0.
\end{equation}
Following the ideas in \citep{crisli,ishii2} let $u$ be an upper semi-continuous (usc) viscosity subsolution of the HJB equation \eqref{hjbinf} and $v$ be a lower semi-continuous (lsc) viscosity supersolution of the same equation in $\mathcal{O}$, such that they are bounded by $U$ and
\begin{equation}\label{boundcompinf}
\limsup_{y'\rightarrow y}u(y')\leq\liminf_{y'\rightarrow y}v(y')\textrm{ for all }y\in\partial \mathcal{O}.
\end{equation}
Define $$v_m(x,p)=v(x,p)+\frac{1}{m}x^2\psi(p)$$
for all $m\geq1$. Then $v_m$ is still lsc and clearly $\beta v_m-\mathcal{A}v_m\geq0$ by definition of $\psi$. Now,
\begin{align*}
\mathcal{M}v_m(x,p)&=\sup_{0\leq\zeta\leq x}v(x-\zeta,\alpha(\zeta,p))+\frac{1}{m}(x-\zeta)^2\psi(\alpha(\zeta,p))+\zeta\alpha(\zeta,p)-k\\
&\leq\sup_{0\leq\zeta\leq x}v(x-\zeta,\alpha(\zeta,p))+\zeta\alpha(\zeta,p)-k+\sup_{0\leq\zeta\leq x}\frac{1}{m}(x-\zeta)^2\psi(\alpha(\zeta,p))\\
&=\mathcal{M}v(x,p)+\frac{1}{m}x^2\psi(p)\\
&\leq v(x,p)+\frac{1}{m}x^2\psi(p)=v_m(x,p).
\end{align*}
Therefore $v_m$ is supersolution of \eqref{hjbinf}. Now, by the growth condition of $u$ and $v$ and equations \eqref{boundinf} and \eqref{psi} we get
\begin{equation}\label{growauxinf}
\lim_{|y|\rightarrow\infty}(u-v_m)(y)=-\infty.
\end{equation}
We will show now that
\begin{equation}\label{compinf}
u\leq v\textrm{ in }\mathcal{O}.
\end{equation}
It is sufficient to show that $\sup\limits_{y\in\bar{\mathcal{O}}}(u-v_m)\leq0$ for all $m\geq1$ since the result is obtained by letting $m\rightarrow\infty$. Suppose that there exists $m\geq1$ such that $\eta=\sup\limits_{y\in\bar{\mathcal{O}}}(u-v_m)>0$. Since $u-v_m$ is usc, by \eqref{growauxinf} and \eqref{boundcompinf} there exist $y_0\in\mathcal{O}$ such that $\eta=(u-v_m)(y_0)$. Let $y_0=(x_0,p_0)$ be the one with minimum norm over all possible maximizers of $u-v_m$. For $i\geq1$, define 
$$\phi_i(y,y')=\frac{i}{2}|y-y'|^4+|y-y_0|^4,$$
$$\Phi_i(y,y')=u(y)-v_m(y')-\phi_i(y,y').$$
Let
$$\eta_i=\sup_{|y|,|y'|\leq|y_0|}\Phi_i(y,y')=\Phi_i(y_i,y'_i).$$
Clearly $\eta_i\geq\eta$. Then, this inequality reads $\frac{i}{2}|y_i-y'_i|^4+|y_i-y_0|^4\leq u(y_i)-v_m(y'_i)-(u-v_m)(y_0)$. Since $|y_i|,|y'_i|\leq |y_0|$ and $u$ and $-v_m$ are bounded above in that region, this implies that $y_i,y'_i\rightarrow y_0$ and $\frac{i}{2}|y_i-y'_i|^4\rightarrow0$ (along a subsequence) as $i\rightarrow\infty$. We also find that $\eta_i\rightarrow\eta$, $u(y_i)-v_m(y'_i)\rightarrow\eta$ and $u(y_i)\rightarrow u(y_0),v_m(y'_i)\rightarrow v(y_0)$. By theorem 3.2 in \citep{crisli}, for all $i\geq1$, there exist symmetric matrices $M_i$ and $M_i'$ such that $(\frac{\partial\phi_i}{\partial y}(y_i,y'_i),M_i)=(d_i,M_i)\in\bar{J}^{2,+}u(y_i)$, $(-\frac{\partial\phi_i}{\partial y'}(y_i,y'_i),M'_i)=(d_i',M_i')\in\bar{J}^{2,-}v_m(y'_i)$ and
$$\begin{pmatrix}
M_i & 0\\
0 & M'_i
\end{pmatrix}\leq D^2\phi_i(y_i,y'_i)+\frac{1}{i}(D^2\phi_i(y_i,y'_i))^2.$$
Since $u$ is a subsolution of \eqref{hjbinf} and $v_m$ is a supersolution, we have
$$\min\{\beta u(y_i)-\mu(p_i)d_{i,2}-\frac{1}{2}\sigma(p_i)^2M_{i,22},u(y_i)-\mathcal{M}u(y_i)\}\leq0,$$
and
$$\min\{\beta v_m(y'_i)-\mu(p'_i)d'_{i,2}-\frac{1}{2}\sigma(p'_i)^2M'_{i,22},v_m(y'_i)-\mathcal{M}v_m(y'_i)\}\geq0.$$
Now, if we show that for infinitely many $i$'s we  have that
\begin{equation}\label{claiminf}
\beta u(y_i)-\mu(p_i)d_{i,2}-\frac{1}{2}\sigma(p_i)^2M_{i,22}\leq0,
\end{equation}
and since it is always true that
$$\beta v_m(y'_i)-\mu(p'_i)d'_{i,2}-\frac{1}{2}\sigma(p'_i)^2M'_{i,22}\geq0,$$
we have that $u\leq v_m$ by following the classical comparison proof in \citep{crisli}. Suppose then, that there exists $i_0$ such that \eqref{claiminf} is not true for all $i\geq i_0$, then for $i\geq i_0$
$$u(y_i)-\mathcal{M}u(y_i)\leq0.$$
Since $v_m$ is a supersolution, we must have that
$$v_m(y'_i)-\mathcal{M}v_m(y'_i)\geq0.$$
Since $u$ is usc, there exist $\zeta_i$ such that $\mathcal{M}u(y_i)=u(x_i-\zeta_i,\alpha(\zeta_i,p_i))+\zeta_i\alpha(\zeta_i,p_i)-k$. Then
$$u(y_i)\leq u(x_i-\zeta_i,\alpha(\zeta_i,p_i))+\zeta_i\alpha(\zeta_i,p_i)-k.$$
Extracting a subsequence if necessary, we assume that $\zeta_i\rightarrow \zeta_0$ as $i\rightarrow\infty$. First, consider $\zeta_0=0$, then by taking $\limsup$ in the inequality above we get $u(y_0)\leq u(y_0) -k$. This is a contradiction since $k>0$. Now assume that $\zeta_0\neq0$. From the above inequalities we have that
$$u(y_i)-v_m(y'_i)\leq u(x_i-\zeta_i,\alpha(\zeta_i,p_i))+\zeta_i\alpha(\zeta_i,p_i)-v_m(x'_i-\zeta'_i,\alpha(\zeta'_i,p'_i))-\zeta'_i\alpha(\zeta'_i,p'_i),$$
for any $0\leq\zeta'_i\leq p'_i$. Since $p'_i\rightarrow p_0$, let $\zeta'_i\rightarrow\zeta_0$ and taking $\limsup$ in the above inequality we get that 
$$\eta\leq (u-v_m)(x_0-\zeta_0,\alpha(\zeta_0,p_0)).$$
This is a contradiction since $y_0$ was chosen with minimum norm among maximizers of $u-v_m$ and $\zeta_0>0$. Therefore \eqref{claiminf} must hold for infinitely many $i$'s and \eqref{compinf} holds. As usual continuity in $\mathcal{O}$ and uniqueness of $V$ follow from the fact that $V$ is a viscosity solution of \eqref{hjbinf}.\\

We have just proved the following theorem:
\begin{theorem}\label{uniq}
Assume condition $\eqref{psi}$ and that the transaction cost $k>0$. If $W$ is a viscosity solution of equation \eqref{hjbinf} that is bounded by $U$ and satisfies the same boundary conditions as $V$, then $W=V$. Furthermore, $V$ is continuous in $\mathcal{O}$.
\end{theorem}

\begin{remark}
Condition \eqref{psi} is satisfied by It\^{o} processes like Brownian Motion, Geometric Brownian Motion, Ornstein-Uhlenbeck and Cox-Ingersoll-Ross.
\end{remark}

\section{No transaction cost}
\setcounter{equation}{0}

From the proof of the above uniqueness result, we can see that the result depends on the fact that $k>0$. Let's start by pointing out that in this case the intervention operator becomes
\begin{equation}\label{subalways}
\mathcal{M}\varphi(y)=\sup_{0\leq\zeta\leq x}\varphi(\Gamma(y,\zeta))+\zeta\alpha(\zeta,p)\geq\varphi(\Gamma(y,0))=\varphi(y),
\end{equation}
for any measurable function $\varphi$. This implies in particular that any measurable function is a viscosity subsolution of \eqref{hjbinf}. On the other hand, $V\geq\mathcal{M}V$ for the value function. Then we have that
$$V\geq\mathcal{M}V\geq V.$$
Assume now that $V\in C^1(\mathcal{O})$. Since $\zeta=0$ is a maximum for $\zeta\mapsto V(\Gamma(y,\zeta))+\zeta\alpha(\zeta,p)$, then for all $y\in\mathcal{O}$:
\begin{align*}
0&\geq\left.\frac{\partial\alpha}{\partial\zeta}(\zeta,p)\frac{\partial V}{\partial p}(y)-\frac{\partial V}{\partial x}(y)+\alpha(\zeta,p)+\zeta\frac{\partial\alpha}{\partial\zeta}(\zeta,p)\right|_{\zeta=0}\\
&=\frac{\partial\alpha}{\partial\zeta}(0,p)\frac{\partial V}{\partial p}(y)-\frac{\partial V}{\partial x}(y)+p.
\end{align*}
Recall that $\alpha$ is non-increasing in $\zeta$, so we define
\begin{equation}\label{gamma}
\gamma(p)=-\frac{\partial\alpha}{\partial\zeta}(0,p),
\end{equation}
for all $p\geq0$. Hence, we get the following condition for $V$:
\begin{equation}\label{singcond}
-\gamma(p)\frac{\partial V}{\partial p}(y)-\frac{\partial V}{\partial x}(y)+p\leq0.
\end{equation}
This suggests that if we assume no fixed transaction cost we should look at a different HJB equation, that is
\begin{equation}\label{hjbinf0}
\min\left\{\beta\varphi -A\varphi,\gamma(p)\frac{\partial\varphi}{\partial p}+\frac{\partial\varphi}{\partial x}-p\right\}=0.
\end{equation}
On the other hand, condition \eqref{suma} implies that it is always better to split the orders into smaller orders. Indeed, given $(\zeta,p)\in\mathcal{O}$ and $0\leq \zeta'\leq \zeta$
$$\zeta'\alpha(\zeta',p)+(\zeta-\zeta')\alpha(\zeta,p)=\zeta\alpha(\zeta,p)+(\zeta-\zeta')(\alpha(\zeta',p)-\alpha(\zeta,p))\geq \zeta\alpha(\zeta,p),$$
since $\alpha$ is non-increasing in $\zeta$.

\subsection{Singular control}

In fact, the equation \eqref{hjbinf0} is the associated equation of the following control problem (\citep{oeksendal05acjd}): In this case our admissible controls are of the singular type, that is
$$dX_t=-d\xi_t,$$
where $\xi_{0}=0$, $\xi$ is an adapted, continuous non-decreasing and non-negative process. The price process in this case follows the dynamics
$$dP_t=\mu(P_{t-})dt+\sigma(P_{t-})dB_t-\gamma(P_{t-})d\xi_t,$$
where $\gamma$ (see \eqref{gamma}) is a non-negative smooth function that accounts for the price impact. In order to guarantee the existence and uniqueness of the process $P_t$, we need to also assume that $\gamma$ is a Lipschitz function (\citep{ProtterBook}). Now, the form of the value function $V_0$ changes to
\begin{equation}\label{valuefunctioninf0}
V_0(y)=\sup_{\xi}\mathbb{E}\left[\int_{0}^{\infty}e^{-\beta t}P_td\xi_t\right],
\end{equation}
for all $y\in\bar{\mathcal{O}}$. In this case the appropriate form of the DPP is
\begin{equation}\label{dppinf0}
V_0(y)=\sup_{\xi}\mathbb{E}\left[\int_{0}^{\tau}e^{-\beta s}P_sd\xi_s+e^{-\beta\tau}V_0(Y_{\tau})\right],
\end{equation}
for any stopping time $\tau$. As before, we can define the continuation region as 
$$\mathcal{C}=\{y\in\mathcal{O}:\gamma(p)\frac{\partial\varphi}{\partial p}+\frac{\partial\varphi}{\partial x}-p>0\}$$
and the trade region as
$$\mathcal{T}=\{y\in\mathcal{O}:\gamma(p)\frac{\partial\varphi}{\partial p}+\frac{\partial\varphi}{\partial x}-p=0\}.$$
Typically, singular controls are allowed to be c\`{a}dl\`{a}g instead of continuous. We decide to restrict our controls for two reasons: (1) Under the absence of fixed transaction cost, the investor will divide the orders into very small pieces as shown above. (2) When the singular control is discontinuous the stochastic integral may not be properly defined  (see \citep{ProtterBook}).

\subsection{Viscosity solution}

Although we only consider continuous strategies, the value function is still a viscosity solution of equation \eqref{hjbinf0} (which definition is similar to \ref{defviscosity}).

\begin{theorem}\label{viscosity0}
The value function $V_0$ defined by \eqref{valuefunctioninf0} is a viscosity solution of \eqref{hjbinf0} in $\mathcal{O}$.
\end{theorem}
\begin{proof} 
Since we can approach finite variation functions by simple functions, by proposition \ref{growprop} we have that 
\begin{equation}\label{singgrow}
V_0\leq U.
\end{equation}
Therefore, $V_0$ is locally bounded.

Subsolution property: Let $y_0\in\mathcal{O}$ and $\varphi\in C^{2}(\mathcal{O})$ such that $y_0$ is a maximizer of $V_0^*-\varphi$ on $\mathcal{O}$ with $V_0^*(y_0)=\varphi(y_0)$. Now suppose that there exists $\kappa>0$ and $\delta>0$ such that
\begin{equation}\label{contr2inf0}
-\beta\varphi(y)+A\varphi(y)\leq-\kappa\textrm{ and }p-\gamma(p)\frac{\partial\varphi}{\partial p}(y)-\frac{\partial\varphi}{\partial x}(y)\leq-\kappa
\end{equation}
for all $y\in\mathcal{O}$ such that $|y-y_0|<\delta$. Let $(y_n)$ be a sequence in $\mathcal{O}$ such that $y_n\rightarrow y_0$ and
$$\lim_{n\rightarrow\infty}V_0(y_n)=V_0^*(y_0).$$
Given any stopping time $\tau$, by \eqref{dppinf0}, for all $n\geq1$ there exists an admissible control $\xi^n$ such that
$$V_0(y_n)\leq\mathbb{E}\left[\int_0^{\tau}e^{-\beta s}P^n_{s}d\xi^n_s+e^{-\beta\tau}V_0(Y^n_{\tau})\right]+\frac{1}{n},$$
where $Y_s^n$ is the process controlled by $\xi^n$ for $s\geq 0$ starting at $y_n$. Since $V_0\leq V_0^*\leq\varphi$, using Dynkin's formula for semimartingales (\citep{ProtterBook}) we have that
\begin{align*}
V_0(y_n)&\leq\mathbb{E}\left[\int_{0}^{\tau}e^{-\beta s}P^n_{s}d\xi^n_s\right]+\varphi(y_n)+\mathbb{E}\left[\int_{0}^{\tau}e^{-\beta s}\left (-\beta \varphi(Y_s^n)+A\varphi(Y_s^n)\right)ds\right]\\
&-\mathbb{E}\left[\int_{0}^{\tau}e^{-\beta s}\left (\gamma(P_s^n)\frac{\partial\varphi}{\partial p}(Y^n_s)+\frac{\partial\varphi}{\partial x}(Y^n_s)\right)d\xi^n_s\right]+\frac{1}{n}.
\end{align*}
Consider again the stopping time
$$\tau_n=\inf\{s\geq 0:|Y_s^n-y_0|\geq\delta\},$$
then by \eqref{contr2inf0}
\begin{align*}
V_0(y_n)&\leq-\kappa\mathbb{E}\left[\int_0^{\tau_n}e^{-\beta s} (ds+d\xi^n_s)\right]+\varphi(y_n)+\frac{1}{n}.
\end{align*}
Taking $n\rightarrow\infty$ we obtain a contradiction since the integral inside the expectation is bounded away from 0 for any admissible control $\xi$  by the a.s continuity of the process $Y^n_s$. Hence at least one of the inequalities in \eqref{contr2inf0} is not possible and this establishes the subsolution property.

Supersolution property: Let $y_0\in\mathcal{O}$ and $\varphi\in C^{2}(\mathcal{O})$ such that $y_0$ is a minimizer of $V_{0*}-\varphi$ on $\mathcal{O}$ with $V_{0*}(y_0)=\varphi(y_0)$. Let $(y_n)$ be a sequence in $\mathcal{O}$ such that $y_n\rightarrow y_0$ and
$$\lim_{n\rightarrow\infty}V_0(y_n)=V_{0*}(y_0).$$
First, suppose that there exists $\theta>0$ and $\delta>0$ such that
\begin{equation}\label{contr1inf0}
\beta\varphi(y)-A\varphi(y)\leq-\theta
\end{equation}
for all $y\in\mathcal{O}$ such that $|y-y_0|<\delta$. Fix $n$ large enough such that $|y_n-y_0|<\delta$ and consider the process $Y_s^n$ for $s\geq 0$ with no intervation, i.e. $\xi=0$, such that $Y^n_{0}=y_n$. Let 
$$\tau_n=\inf\{s\geq 0:|Y_s^n-y_0|\geq\delta\}.$$
Now, by Dynkin's formula for semimartingales  and \eqref{contr1inf0} we have 
\begin{align*}
\mathbb{E}[e^{-\beta\tau_n}\varphi(Y^n_{\tau_n})]&=\varphi(y_n)+\mathbb{E}\left[\int_{0}^{\tau_n}e^{-\beta s}\left (-\beta \varphi(Y_s^n)+A\varphi(Y_s^n)\right)ds\right]\\
&-\mathbb{E}\left[\int_{0}^{\tau_n}e^{-\beta s}\left (\gamma(P_s^n)\frac{\partial\varphi}{\partial p}(Y^n_s)+\frac{\partial\varphi}{\partial x}(Y^n_s)\right)d\xi_s\right]\\
&=\varphi(y_n)+\mathbb{E}\left[\int_{0}^{\tau_n}e^{-\beta s}\left (-\beta \varphi(Y_s^n)+A\varphi(Y_s^n)\right)ds\right]\\
&\geq\varphi(y_n)-\theta\mathbb{E}\left[\int_{0}^{\tau_n}e^{-\beta s}ds\right].
\end{align*}
As before, from here we can draw a contradiction with $V_{0*}(y_0)=\varphi(y_0)$ by the a.s. continuity if the process $Y^n_s$. 
Now, take $h>0$ and consider the process $Y_t$ with control process $d\xi_t=\frac{1}{h} 1_{[0,h]}(t)dt$ and $Y_0=y$ for given $y\in\mathcal{O}$. Using \eqref{dppinf0} we can show that
\begin{align*}
V_0(y)&\geq\mathbb{E}\left[\int_{0}^{h}e^{-\beta s}P_sd\xi_s+e^{-\beta h}V(Y_h)\right]\\
&\geq\mathbb{E}\left[\int_{0}^{h}e^{-\beta s}P_sd\xi_s+e^{-\beta h}\varphi(Y_h)\right]\\
&=\mathbb{E}\left[\frac{1}{h}\int_{0}^{h}e^{-\beta s}P_sds+e^{-\beta h}\varphi(Y_h)\right].
\end{align*}
By Dynkin's formula again,
\begin{align*}
\mathbb{E}[e^{-\beta h}\varphi(Y_h)]&=\varphi(y)+\mathbb{E}\left[\int_{0}^he^{-\beta s}\left (-\beta \varphi(Y_s)+A\varphi(Y_s)\right)ds\right]\\
&-\mathbb{E}\left[\int_{0}^he^{-\beta s}\left (\gamma(P_s)\frac{\partial\varphi}{\partial p}(Y_s)+\frac{\partial\varphi}{\partial x}(Y_s)\right)d\xi_s\right]\\
&=\varphi(y)+\mathbb{E}\left[\int_{0}^he^{-\beta s}\left (-\beta \varphi(Y_s)+A\varphi(Y_s)\right)ds\right]\\
&-\frac{1}{h}\mathbb{E}\left[\int_{0}^he^{-\beta s}\left (\gamma(P_s)\frac{\partial\varphi}{\partial p}(Y_s)+\frac{\partial\varphi}{\partial x}(Y_s)\right)ds\right].
\end{align*}
Letting $h\rightarrow0$, we have
$$V_0(y)\geq\varphi(y)+p-\gamma(p)\frac{\partial\varphi}{\partial p}(y)-\frac{\partial\varphi}{\partial x}(y).$$
Therefore, for all $n\geq1$ we have 
$$V_0(y_n)\geq\varphi(y_n)+p_n-\gamma(p_n)\frac{\partial\varphi}{\partial p}(y_n)-\frac{\partial\varphi}{\partial x}(y_n).$$
Since $\gamma$ is continuous, letting $n\rightarrow\infty$ we get
$$\varphi(y_0)=V_{0*}(y_0)\geq\varphi(y_0)+p_0-\gamma(p_0)\frac{\partial\varphi}{\partial p}(y_0)-\frac{\partial\varphi}{\partial x}(y_0)$$
as desired. This establishes the supersolution property.
\end{proof}

\subsection{Uniqueness}

Recall that with the impulse formulation we do not have uniqueness in absence of transaction cost. This is not the case with the singular control formulation.

\begin{theorem}\label{uniq0}
Assume that  \eqref{psi} is satisfied. If $W$ is a viscosity solution of equation \eqref{hjbinf0} that is bounded by $U$ and satisfies the same boundary conditions as $V_0$, then $W=V_0$. Furthermore, $V_0$ is continuous in $\mathcal{O}$.
\end{theorem}

\begin{proof}
The proof follows the same strategy as in the impulse control case. Let $u$ be an upper semi-continuous (usc) viscosity subsolution of the HJB equation \eqref{hjbinf0} and $v$ be a lower semi-continuous (lsc) viscosity supersolution of the same equation in $\mathcal{O}$, such that they are bounded by $U$ and condition \eqref{boundcompinf} holds. Define 
$$v_m(x,p)=\left(1-\frac{1}{m}\right)v(x,p)+\frac{1}{m}\left(C(x+1)^2\psi(p)+1\right)$$
for all $m\geq1$ and $C$ as in \eqref{boundinf}. Recall that $\gamma$ is non-negative and $\psi$ is an increasing function, then \eqref{boundinf} implies that
\begin{align*}
-p+\frac{\partial v_m}{\partial x}+\gamma(p)\frac{\partial v_m}{\partial p}&\geq
-p+\left(1-\frac{1}{m}\right)p+\frac{\partial}{\partial x}\frac{1}{m}C(x+1)^2\psi(p)+\gamma(p)\frac{\partial}{\partial p}\frac{1}{m}C(x+1)^2\psi(p)\\
&=-\frac{1}{m}p+\frac{1}{m}2C(x+1)\psi(p)+\gamma(p)\frac{1}{m}C(x+1)^2\psi'(p)\\
&\geq-\frac{1}{m}p+\frac{2}{m}p(x+1)+\gamma(p)\frac{1}{m}C(x+1)^2\psi'(p)\\
&\geq \frac{1}{m}p.
\end{align*}
Also $(\beta I-A)\left(\frac{1}{m}\right)=\frac{\beta}{m}>0$, where $I$ is the identity operator. Therefore $v_m$ is a strict supersolution of \eqref{hjbinf0} in $\mathcal{O}$. Following the same lines and definitions as in the previous proof we have
$$\min\{\beta u(y_i)-\mu(p_i)d_{i,2}-\frac{1}{2}\sigma(p_i)^2M_{i,22},-p_i+d_{i,1}+\gamma(p_i)d_{i,2}\}\leq0,$$
and
$$\min\{\beta v_m(y'_i)-\mu(p'_i)d'_{i,2}-\frac{1}{2}\sigma(p'_i)^2M'_{i,22},-p'_i+d'_{i,1}+\gamma(p'_i)d'_{i,2}\}\geq\delta_i,$$
where $\delta_i=\min\left\{\frac{p'_i}{m},\frac{\beta}{m}\right\}$. Since $p'_i\rightarrow p_0$ and $y_0\in\mathcal{O}$, $\delta_i>0$ for large enough $i$. 
We need to show now that for infinitely many $i$'s we have that
\begin{equation}\label{claiminf0}
\beta u(y_i)-\mu(p_i)d_{i,2}-\frac{1}{2}\sigma(p_i)^2M_{i,22}\leq0.
\end{equation}
Suppose then, that there exists $i_0$ such that \eqref{claiminf0} is not true for all $i\geq i_0$, then for $i\geq i_0$
$$-p_i+d_{i,1}+\gamma(p_i)d_{i,2}\leq0.$$
Since $v_m$ is a supersolution, we must have that
$$-p'_i+d'_{i,1}+\gamma(p'_i)d'_{i,2}\geq\delta_i.$$
Hence,
$$p_i-p_i'-(d_{i,1}-d'_{i,1})-(\gamma(p_i)d_{i,2}-\gamma(p'_i)d'_{i,2})\geq\delta_i.$$
Since $d_i,d'_i$ goes to 0 as $i$ goes to $\infty$, we get the contradiction $0\geq\delta_0=\min\left\{\frac{p_0}{m},\frac{\beta}{m}\right\}>0$. Therefore \eqref{claiminf0} must hold for infinitely many $i$'s and the comparison result holds. Everything follows now as before.
\end{proof}

\subsection{Optimal strategy for a special case}
Previous sections characterize the value function of our problem in different formulations. We will calculate now the explicit solution of the value function and describe the optimal strategy in a particular case. Let us come back to the impulse control case. Since we are allowed to do multiple trades at the same time, we are going to explore this strategy. Assumption \eqref{suma} guarantees that the price impact does not change by splitting the trades, but the profit obtained by doing so could be greater.  Let's define the following function 
\begin{equation}\label{w}
W(y)=\int_0^x\alpha(s,p)ds\textrm{  for  }y\in\mathcal{O}.
\end{equation}
This is the best that we can do when we do many trades at the same time. It is clear that this is not attainable with any impulse control. Since $\alpha$ is non-increasing on $x$ and positive, we have for all $y\in\mathcal{O}$
\begin{equation}\label{inqlema1}
x\alpha(x,p)\leq\int_0^x\alpha(s,p)ds.
\end{equation}
Therefore, for all $0\leq\zeta\leq x$
\begin{align*}
W(\Gamma(y,\zeta))+\zeta \alpha(\zeta,p)&= \zeta \alpha(\zeta,p)+\int_0^{x-\zeta}\alpha(s,\alpha(\zeta,p))ds\\
&\leq \zeta \alpha(\zeta,p)+\int_0^{x}\alpha(s,p)ds-\int_0^{\zeta}\alpha(s,p)ds\\
&= W(y)+\zeta \alpha(\zeta,p)-\int_0^{\zeta}\alpha(s,p)ds\\
&\leq W(y),
\end{align*}
where the last inequality follows from \eqref{inqlema1}. Hence $\mathcal{M}W\leq W$ and therefore $\mathcal{M}W= W$ by \eqref{subalways}. On the other hand, the function $W$ satisfies \eqref{singcond} with equality. Indeed, by the condition \eqref{suma} we have that for any $\zeta_1$, $\zeta_2$ and $p$
$$\frac{\partial\alpha}{\partial\zeta}(\zeta_1+\zeta_2,p)=\frac{\partial\alpha}{\partial p}(\zeta_1,\alpha(\zeta_2,p))\frac{\partial\alpha}{\partial\zeta}(\zeta_2,p),$$
and taking $\zeta_2=0$ we obtain
$$\frac{\partial\alpha}{\partial\zeta}(\zeta_1,p)=\frac{\partial\alpha}{\partial p}(\zeta_1,p)\frac{\partial\alpha}{\partial\zeta}(0,p)=-\gamma(p)\frac{\partial\alpha}{\partial p}(\zeta_1,p).$$
Now, since $\alpha$ is smooth we find
\begin{align*}
-\gamma(p)\frac{\partial W}{\partial p}(y)-\frac{\partial W}{\partial x}(y)+p&=-\gamma(p)\int_0^x\frac{\partial\alpha}{\partial p}(s,p)ds-\frac{\partial }{\partial x}\int_0^x\alpha(s,p)ds+p\\
&=\int_0^x\frac{\partial\alpha}{\partial\zeta}(s,p)ds-\alpha(x,p)+p\\
&=\alpha(x,p)-\alpha(0,p)-\alpha(x,p)+p=0.
\end{align*}
If we had also that $\beta W-AW\geq0$, then $W$ would solve both equations \eqref{hjbinf} and \eqref{hjbinf0} and $\mathcal{T}=\mathcal{O}$. 

Now, \citep{ajay} considers impact functions of the form $\alpha(x,p)=pc(x)$, where $0\leq c\leq1$ is nonincreasing. In our case, by condition \eqref{suma}, $c$ must satisfy $c(x_1)c(x_2)=c(x_1+x_2)$ and therefore we end up with the following price impact functions and $W$:
\begin{align}
\alpha(x,p)&=pe^{-\lambda x}\label{alfa}\\
\gamma(p)&=\lambda p\\
W(x,p)&=\frac{p}{\lambda}(1-e^{-\lambda x})\label{ww}
\end{align}
with $\lambda>0$. This function was proposed also in \citep{He} and \citep{Ly}. Let's consider this price impact function for the moment. In this case we have the following:
\begin{theorem}\label{conecspecial}
$V_0=W=V$ if and only if $U(x,p)=xp$.
\end{theorem}
\begin{proof}
If $V_0=W$ then  $\beta W-AW\geq0$ and therefore $\beta\varphi-A\varphi\geq0$ for $\varphi(p)=p$. By the uniqueness result for optimal stopping problems (see Theorem 3.1 in \citep{oksenreik})
$$p=\sup\limits_{\tau}\mathbb{E}[e^{-\beta\tau}P_{\tau}],$$
that is $U(x,p)=xp$.
Suppose that 
$$U(x,p)=x\sup\limits_{\tau}\mathbb{E}[e^{-\beta\tau}P_{\tau}]=xp,$$ 
for $y\in\mathcal{O}$. This means that $\beta\varphi-A\varphi\geq0$ for $\phi(p)=p$. Therefore $\beta W-AW\geq0$ and $W$ satisfies the HJB equation \eqref{hjbinf0} with $\mathcal{T}=\mathcal{O}$. Also, $W$ satisfies the growth condition and has the same boundary conditions as $V_0$ by \eqref{optstop}. By Theorem \ref{uniq0}, we have that $W=V_0$. To prove the second equality we will do induction in the number of trades. Note that the function $\zeta\mapsto \zeta e^{-\lambda\zeta}$ in $[0,x]$ attains its maximum at $\hat{x}=\min\{x,\frac{1}{\lambda}\}$. Then,
$$\sup\limits_{\nu\in\Upsilon_1}\mathbb{E}[e^{-\beta \tau_1}\zeta_1P_{\tau_1-}e^{-\lambda\zeta_1}]\leq U(\hat{x},p)=\hat{x}p\leq W(x,p).$$
Now, let $\nu\in\Upsilon_n$. Hence,
$$\mathbb{E}[e^{-\beta\tau_1}\zeta_1P_{\tau_1}]=\mathbb{E}\left[e^{-\beta\tau_1}\mathbb{E}[\zeta_1P_{\tau_1-}e^{-\lambda\zeta_1}|\mathcal{F}_{\tau_1}]\right].$$
On the other hand, by induction hypothesis we have
\begin{align*}
\mathbb{E}\left[e^{-\beta\tau_1}\sum\limits_{i=2}^ne^{-\beta(\tau_i-\tau_1)}\zeta_iP_{\tau_i}\right]&
=\mathbb{E}\left[e^{-\beta\tau_1}\mathbb{E}\left[\left.\sum\limits_{i=2}^n e^{-\beta(\tau_i-\tau_1)}\zeta_iP_{\tau_i}\right|\mathcal{F}_{\tau_1}\right]\right]\\
&\leq\mathbb{E}\left[e^{-\beta\tau_1}\mathbb{E}[V(x-\zeta_1,e^{-\lambda\zeta_1}P_{\tau_1-})|\mathcal{F}_{\tau_1}]\right]\\
&\leq\mathbb{E}\left[e^{-\beta\tau_1}\mathbb{E}[W(x-\zeta_1,e^{-\lambda\zeta_1}P_{\tau_1-})|\mathcal{F}_{\tau_1}]\right]
\end{align*}
Combining both inequalities above we have
$$\mathbb{E}\left[\sum\limits_{i=1}^n e^{-\beta\tau_i}\zeta_iP^{\nu}_{\tau_i}\right]\leq\mathbb{E}[e^{-\beta\tau_1}W(x,P_{\tau_1-})]\leq W(x,p).$$
Again, by Lemma 7.1 in \citep{oeksendal05acjd}, the left hand side converges to $V$ as $n\rightarrow\infty$. Clearly the other inequality holds and the proof is complete.
\end{proof}

\paragraph{Example} Consider the case where the price process is a geometric Brownian motion. This is the only process that is considered in the papers \citep{ajay,He,Ly}. The unperturbed price process is
$$dP_t=\mu P_{t}dt+\sigma P_{t}dB_t,$$
with $\sigma>0$. It is easy to see that the value function $U$ is finite if and only if $\beta>\mu$. In this case the function $\psi$ takes the form
$$\psi(p)=p^{\nu},$$
where $\nu>1$, therefore condition \eqref{psi} holds. Now, the condition \eqref{optstop} reads
$$0\leq V(x,p)\leq U(x,p)=xp.$$
This implies that $V_0=V=W$. We can see that in this case the value function $W$ is not attainable with any impulse control, but we can approach it by trading smaller and smaller orders. We will show now how we can approach $W$ with singular (in fact regular) controls. Let $u>0$ and consider the strategy $d\xi_t=udt$, that is, selling shares at a constant speed $u$ until the investor executes the position. Then, 
$$P_t=p\exp\{(\mu-\lambda u-\frac{1}{2}\sigma^2)t+\sigma B_t\}$$
and 
\begin{align*}
\mathbb{E}\left[\int_0^{\infty}e^{-\beta t}P_td\xi_t\right]&=u\mathbb{E}\left[\int_0^{x/u}e^{-\beta t}P_tdt\right]\\
&=u\int_0^{x/u}e^{-\beta t}\mathbb{E}[P_t]dt\\
&=up\int_0^{x/u}e^{(\mu-\lambda u-\beta)t}dt\\
&=\frac{pu}{\mu-\lambda u-\beta}\left(e^{(\mu-\lambda u-\beta)x/u}-1\right)
\end{align*}
by using Fubini's theorem since the integrand is positive. Taking $u\rightarrow\infty$ this expression converges to $W$.

\section{Connection between both formulation}
\setcounter{equation}{0}

Theorem \ref{conecspecial} shows that $V=V_0$ for a special case, i.e., the value function of two different problems are the same. We are going to show that this is not a coincidence. Let us start with some notation: Given $k\geq0$ and $y=(x,p)\in\bar{\mathcal{O}}$ we denote:
 $$V^{(k)}(y)=\sup_{\nu}\mathbb{E}\left[\sum\limits_{n=1}^{M}e^{-\beta\tau_n}(\zeta_nP_{\tau_n}-k)\right]$$
and
$$\mathcal{M}^{(k)}\varphi(y)=\sup_{0\leq\zeta\leq x}\varphi(\Gamma(y,\zeta))+\zeta\alpha(\zeta,p)-k.$$
\begin{lemma}\label{limit}
For all $y\in \bar{\mathcal{O}}$ we have
$$\lim\limits_{k\rightarrow0}V^{(k)}(y)=V^{(0)}(y).$$
\end{lemma}
\begin{proof} It is clear that $V^{(0)}$ is an upper bound. Let $\epsilon>0$, then there is $m\geq0$ and $\nu\in\Upsilon_m$ such that
$$V^{(0)}(y)\leq\mathbb{E}\left[\sum\limits_{i=1}^me^{-\beta\tau_i}\zeta_i\alpha(\zeta_i,P_{\tau_i-})\right]+\epsilon.$$
For any $k\leq\frac{\epsilon}{m}$ we have that
\begin{align*}
V^{(0)}(y)&\leq\mathbb{E}\left[\sum\limits_{i=1}^me^{-\beta\tau_i}\left(\zeta_i\alpha(\zeta_i,P_{\tau_i-})-k\right)\right]+k\mathbb{E}\left[\sum\limits_{i=1}^me^{-\beta\tau_i}\right]+\epsilon\\
&\leq V^{(k)}(y)+2\epsilon.
\end{align*}
\end{proof}

\begin{theorem}
$V^{(0)}$ solves the Hamilton-Jacobi-Bellman
\begin{equation*}
\min\left\{\beta\varphi -A\varphi,\gamma(p)\frac{\partial\varphi}{\partial p}+\frac{\partial\varphi}{\partial x}-p\right\}=0\tag{\ref{hjbinf0}},
\end{equation*}
with $\gamma$ as in \eqref{gamma}. Therefore, by theorem \ref{uniq0}, $V^{(0)}=V_0$. 
\end{theorem}
\begin{proof}
First, consider the case when there is no impact in the price. Then $\gamma\equiv0$ and by proposition \ref{growprop} $V^{(0)}(y)=xU^{(0)}(p)$, where $U^{(0)}(p)=\sup_{\tau}\mathbb{E}[e^{-\beta\tau}P_{\tau}]$. Since $U$ is viscosity solution of
$$\min\left\{\beta\varphi -A\varphi,\varphi-p\right\}=0$$
Then $V^{(0)}$ is solution of \eqref{hjbinf0}.

Assume now that there is price impact. We know that $V^{(0)}$ satisfies the equation
 $$\min\left\{\beta\varphi-A\varphi,\varphi-\mathcal{M}^{(0)}\varphi\right\}=0.$$
Supersolution property: Let $y_0\in\mathcal{O}$ and $\varphi\in C^{2}(\mathcal{O})$ such that $y_0$ is a minimizer of $V^{(0)}_*-\varphi$ on $\mathcal{O}$ with $V_*^{0}(y_0)=\varphi(y_0)$. Hence, $\beta V_*^{(0)}(y_0)-A\varphi(y_0)\geq0$. Now, let $0\leq\zeta^*\leq x_0$ such that $\mathcal{M}^{(0)}\varphi(y_0)=\varphi(x_0-\zeta^*,\alpha(\zeta^*,p_0))+\zeta^*\alpha(\zeta^*,p_0)$. Thus,
\begin{align*}
0&\leq V^{(0)}_*(y_0)-\mathcal{M}^{(0)}V^{(0)}_*(y_0)\\
&\leq V^{(0)}_*(y_0)- V^{(0)}_*(x_0-\zeta^*,\alpha(\zeta^*,p_0))-\zeta^*\alpha(\zeta^*,p_0)\\
& \leq \varphi(y_0)-\varphi(x_0-\zeta^*,\alpha(\zeta^*,p_0))-\zeta^*\alpha(\zeta^*,p_0)\\
&=\varphi(y_0)-\mathcal{M}^{(0)}\varphi(y_0).
\end{align*}
Since $\varphi\leq\mathcal{M}^{(0)}\varphi$, then $\varphi(y_0)=\mathcal{M}^{(0)}\varphi(y_0)$. This implies that $\zeta=0$ is a maximum for $\zeta\mapsto \varphi(x_0-\zeta,\alpha(\zeta,p_0))+\zeta\alpha(\zeta,p_0)$, therefore
\begin{align*}
0&\geq\left.\frac{\partial\alpha}{\partial\zeta}(\zeta,p_0)\frac{\partial \varphi}{\partial p}(y_0)-\frac{\partial \varphi}{\partial x}(y_0)+\alpha(\zeta,p_0)+\zeta\frac{\partial\alpha}{\partial\zeta}(\zeta,p_0)\right|_{\zeta=0}\\
&=-\gamma(p_0)\frac{\partial \varphi}{\partial p}(y_0)-\frac{\partial \varphi}{\partial x}(y_0)+p_0.
\end{align*}
Subsolution property: Let $y_0\in\mathcal{O}$ and $\varphi\in C^{2}(\mathcal{O})$ such that $y_0$ is a maximizer of $V^{(0)*}-\varphi$ on $\mathcal{O}$ with $V^{(0)*}(y_0)=\varphi(y_0)$. Without loss of generality we can assume that $y_0$ is a strict local maximum, that is, there exists $\delta>0$ such that $y_0$ is maximum of $V^{(0)*}-\varphi$ over $B_{\delta}(y_0)\subset\mathcal{O}$. Let $(y_n)$ be a sequence in $B_{\delta}(y_0)$ such that $y_n\rightarrow y_0$ and 
$$\lim_{n\rightarrow\infty}V^{(0)}(y_n)=V^{(0)*}(y_0).$$
Recall that $V^{(k)}$ is continuous and is the unique viscosity solution of
\begin{equation}\label{hjbinfk}
\min\left\{\beta\varphi-A\varphi,\varphi-\mathcal{M}^{(k)}\varphi\right\}=0,
\end{equation}
for all $k>0$. Let $y_k$ be a maximum of $V^{(k)}-\varphi$ over  $B_{\delta}(y_0)$ and let $y'$ be a limit point of $(y_k)$ as $k\rightarrow0$. For all $k$ and all $n$ we have that
$$V^{(k)}(y_n)-\varphi(y_n)\leq V^{(k)}(y_k)-\varphi(y_k)\leq V^{(0)}(y_k)-\varphi(y_k).$$
By lemma \ref{limit}, taking $k\rightarrow0$ along the sequence such that $y_k\rightarrow y'$, we have that for all $n$
$$V^{(0)}(y_n)-\varphi(y_n)\leq V^{(0)*}(y')-\varphi(y').$$
Taking $n\rightarrow\infty$ we obtain that
$$V^{(0)*}(y_0)-\varphi(y_0)\leq V^{(0)*}(y')-\varphi(y')$$
and therefore $y'=y_0$ since $y_0$ is a strict local maximum. Thus, $y_k\rightarrow y_0$ as $k\rightarrow0$. Let $\epsilon>0$ , since $V^{(0)}(y_n),\varphi(y_n)\rightarrow\varphi(y_0)$, then for a fix $n$ large enough $\varphi(y_n)-V^{(0)}(y_n)\leq\epsilon$ and for $k$ small enough $V^{(0)}(y_n)-V^{(k)}(y_n)\leq\epsilon$. Hence
$$\varphi(y_k)-V^{(k)}(y_k)\leq\varphi(y_n)-V^{(k)}(y_n)\leq2\epsilon.$$
The above shows that $y_k$ is a local maximum of $V^{(k)}-\varphi^{(k)}$ over $B_{\delta}(y_0)$ where $V^{(k)}(y_k)=\varphi^{(k)}(y_k)$, $\varphi^{(k)}=\varphi-\epsilon_k$ and $0<\epsilon_k\rightarrow0$ as $k\rightarrow0$.
Since $V^{(k)}$ is subsolution of \eqref{hjbinfk}, we can consider two cases:
\begin{itemize}
\item There exists a sequence such that $\beta \varphi^{(k)}(y_k)-A\varphi^{(k)}(y_k)=\beta V^{(k)}(y_k)-A\varphi^{(k)}(y_k)\leq0$. Taking $k\rightarrow0$ along the sequence we have that
$$\beta V^{(0)*}(y_0)-A\varphi(y_0)=\beta\varphi(y_0)-A\varphi(y_0)\leq0$$
and $V^{(0)}$ is subsolution.
\item For all $k$ small enough $\beta \varphi^{(k)}(y_k)-A\varphi^{(k)}(y_k)>0$. This implies that there exists $0\leq\zeta_k\leq x_k$ such that
\begin{equation}\label{zetak}
\mathcal{M}V^{(k)}(y_k)=V^{(k)}(x_k-\zeta_k,\alpha(\zeta_k,p_k))+\zeta_k\alpha(\zeta_k,p_k)-k\geq\varphi^{(k)}(y_k).
\end{equation}
Let $\zeta'$ be a limit point of $(\zeta_k)$ as $k\rightarrow0$. We claim that $\zeta'=0$: Suppose $\zeta'>0$, then for $k$ small enough such that $0<\zeta'/2<\zeta_k$
\begin{align*}
\mathcal{M}V^{(k)}(y_k)&\geq V^{(k)}(x_k-\zeta'/2,\alpha(\zeta'/2,p_k))+\alpha(\zeta'/2,p_k)\zeta'/2-k\\
&\geq V^{(k)}(x_k-\zeta_k,\alpha(\zeta_k,p_k))+\alpha(\zeta'/2,p_k)\zeta'/2+\alpha(\zeta_k,p_k)(\zeta_k-\zeta'/2)-2k\\
&=\mathcal{M}V^{(k)}(y_k)+\zeta'/2[\alpha(\zeta'/2,p_k)-\alpha(\zeta_k,p_k)]-k.
\end{align*}
where the second inequality follows from \eqref{suma} and the definition of $V^{(k)}$. Now, since $\alpha$ is strictly decreasing in $\zeta$ (otherwise condition \eqref{suma} cannot hold) we can choose $k$ small such that $\zeta'/2[\alpha(\zeta'/2,p_k)-\alpha(\zeta_k,p_k)]>k$ and we get a contradiction.

Since $\zeta_k\rightarrow0$, for $k$ small we have that $(x_k-\zeta_k,\alpha(\zeta_k,p_k))\in B_{\delta}(y_0)$. Therefore, from \eqref{zetak} we have that $\varphi^{(k)}(x_k-\zeta_k,\alpha(\zeta_k,p_k))+\zeta_k\alpha(\zeta_k,p_k)-k\geq\varphi^{(k)}(y_k)$, that is the same as
\begin{equation}\label{key}
\varphi(x_k-\zeta_k,\alpha(\zeta_k,p_k))+\zeta_k\alpha(\zeta_k,p_k)>\varphi(x_k-\zeta_k,\alpha(\zeta_k,p_k))+\zeta_k\alpha(\zeta_k,p_k)-k\geq\varphi(y_k).
\end{equation}
If we consider, for each $k$, the map $\zeta\mapsto\varphi(x_k-\zeta,\alpha(\zeta,p_k))+\zeta\alpha(\zeta,p_k)$, \eqref{key} implies that there exists $\hat{\zeta}_k\in(0,\zeta_k)$ such that
$$0\leq\frac{\partial\alpha}{\partial\zeta}(\hat{\zeta}_k,p_k)\frac{\partial \varphi}{\partial p}(x_k-\hat{\zeta}_k,\alpha(\hat{\zeta}_k,p_k))-\frac{\partial \varphi}{\partial x}(x_k-\hat{\zeta}_k,\alpha(\hat{\zeta}_k,p_k))+\alpha(\hat{\zeta}_k,p_k)+\hat{\zeta}_k\frac{\partial\alpha}{\partial\zeta}(\hat{\zeta}_k,p_k).$$
Taking $k\rightarrow0$
$$0\leq-\gamma(p_0)\frac{\partial \varphi}{\partial p}(y_0)-\frac{\partial \varphi}{\partial x}(y_0)+p_0.$$
\end{itemize}
\end{proof}

\section{Computational Examples}

We are going to present different choices of price processes.  Throughout this section we will consider the price impact function:
\begin{align*}
\alpha(x,p)&=pe^{-\lambda x}\\
\gamma(p)&=\lambda p
\end{align*}
In the following examples, analytical solutions for $V$ do not seem easy to find,  so we used an implicit numerical scheme following chapter 6 in \citep{Kushner}. In particular, we used the Gauss-Seidel iteration method for approximation in the value space. Additionally, for the impulse control case we followed the iterative procedure described in \citep{oeksendal05acjd}.

\subsection{Impulse control with positive fixed transaction cost}

    \begin{figure}[h]
        \centering
        \includegraphics[scale=0.4]{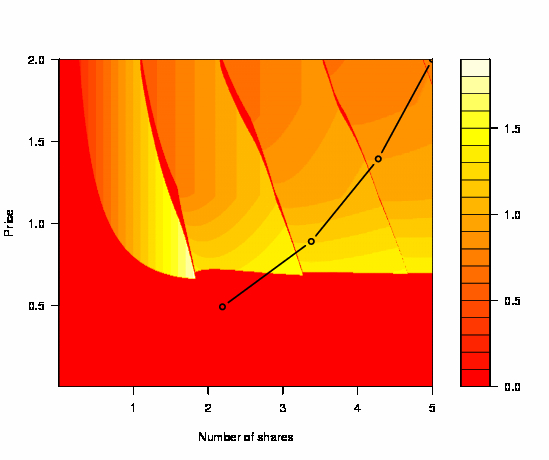}
      \label{traGBM}
      \caption{Optimal number of shares with parameters $\lambda=0.5$, $\mu=2$, $\sigma=1$, $\beta=4$ and $k=0.2$.}
     \end{figure}

Consider the price process following a geometric Brownian motion with $\mu<\beta$ so that the value function is finite. Figure 1 shows the contour plot of the optimal number of shares the investor need to trade. The figure also shows the optimal strategy when the investor starts with 5 shares at a price of 2. At time 0, the investor needs to trade three times until the state variable enters the continuation region $\mathcal{C}$ (i.e. when the optimal number of shares is 0).When $k$ is smaller, the number of trades at time 0 increases and the continuation region shrinks. When $k=0$ we obtain the situation described in theorem \ref{conecspecial}.

\subsection{Singular control}

Consider the case when the price process follows an arithmetic Brownian motion. Then the price dynamics are
$$dP_t=\mu dt+\sigma dB_t-\lambda P_{t}d\xi_t,$$
with $\sigma>0$. In this case the value function is always finite, regardless of $\mu$, due to the exponential decay of the discount factor. Since 0 is an absorbing boundary for this process the boundary conditions are given by \eqref{boundaryinfab}. Figure \ref{valueBM} shows the value function obtained by the scheme.

\begin{figure}
  \centering
  \subfigure[Value function in the BM case with parameters $\lambda=0.5$, $\mu=4$, $\sigma=0.5$ and $\beta=1$.]   
  {
	\includegraphics[scale=0.4]{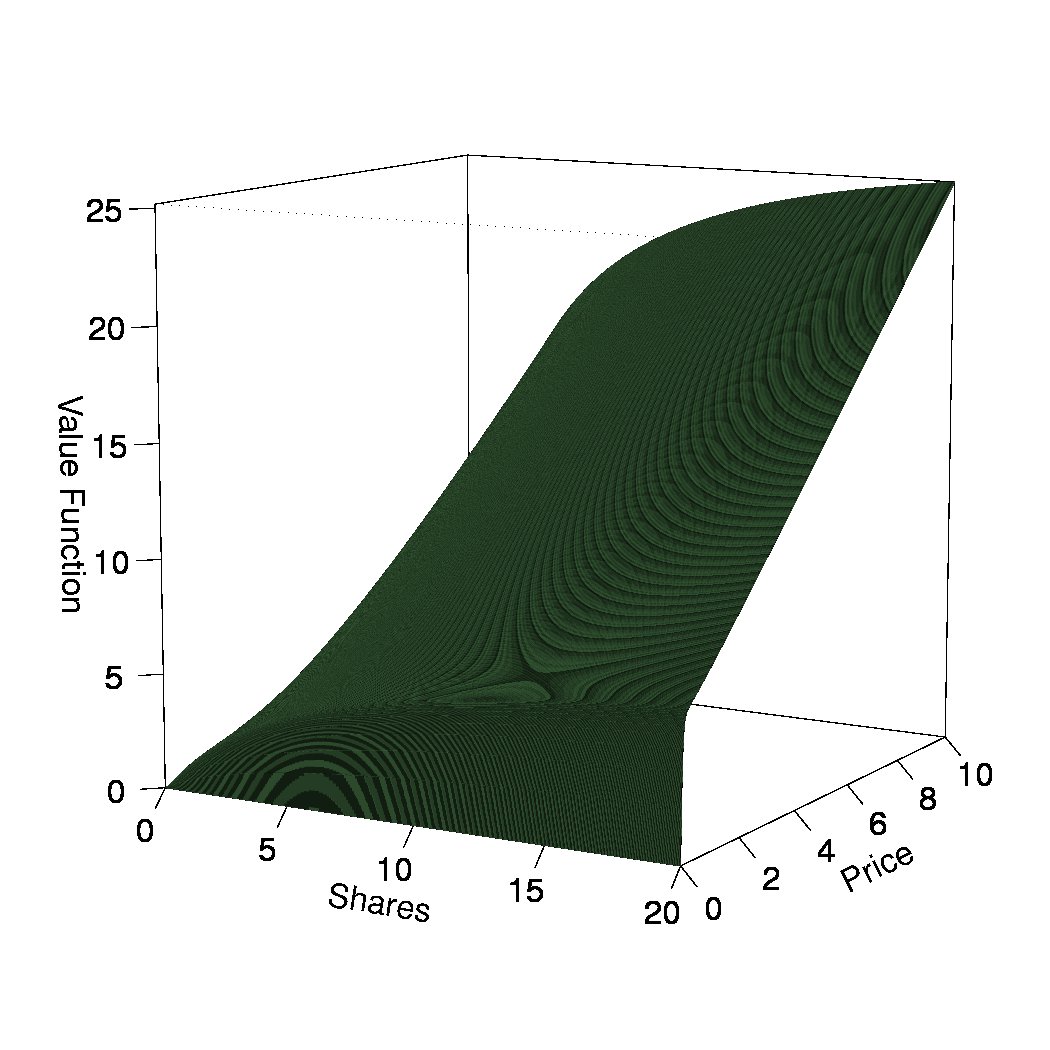}
        \label{valueBM}
  }
  \subfigure[Continuation-trade region in the BM case. The solid line shows the contour with parameters $\lambda=0.5$, $\mu=4$, $\sigma=0.5$ and $\beta=1$. In the other lines only the indicated parameter has been changed.]
  {
  	\includegraphics[scale=0.4]{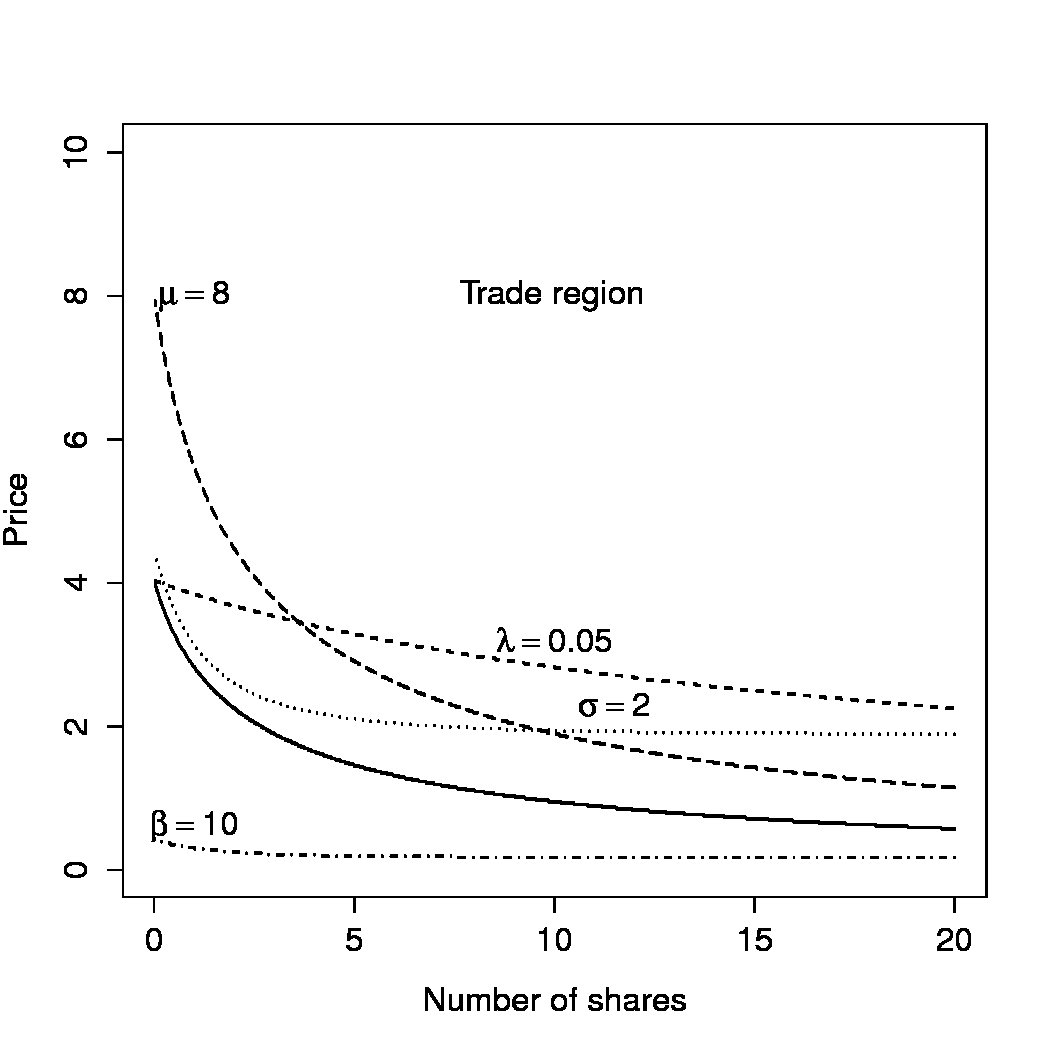}
	\label{tradeBM}	
  }
  \caption{Value function and continuation-trade region in the BM case.}
\end{figure}

First note that the conditions of Theorem \ref{conecspecial} are not satisfied, that is $U(x,p)\neq xp$, and therefore $\mathcal{T}\neq\mathcal{O}$, as shown in figure \ref{tradeBM}. Thus, in this case the optimal strategy would be to trade very fast in the trading region until the state variable hit the free boundary. The figure also shows how the different parameters affect the continuation/trade regions. Now, let's see how the change in the parameters of the model affect the value function $V$. Figure \ref{lambdaBM} shows that the value function is very sensitive to changes in the parameter $\lambda$ for small values but not so much for large values. This behavior is common to both processes GBM (described by theorem \ref{conecspecial}) and BM. This means that the bigger the investor (i.e. the larger the price impact) the less sensitive to small changes in the value of $\lambda$. Clearly the value function decreases as the impact increases.

If $\beta=0$, the value function would not be finite for any $\mu>0$, so small values of $\beta$ yield a very large value of $V$. As $\beta$ increases the effect in $V$ is diminishing. Also, the investor has to act greedily and therefore the trade region approaches to $\mathcal{O}$ and $V$ approaches to $W$.

\begin{figure}
    \centering
    \subfigure[Change in $V(5,2)$ as $\lambda$ varies and $\mu=4$, $\sigma=0.5$ and $\beta=1$.]
    {
        \includegraphics[scale=0.3]{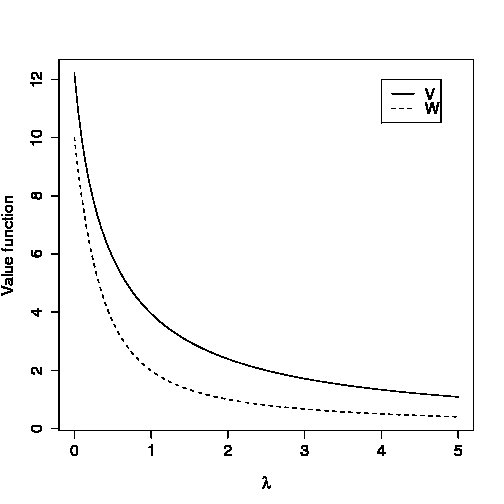}
        \label{lambdaBM}
    }
    \subfigure[Change in $V(5,2)$ as $\beta$ varies and $\mu=4$, $\sigma=0.5$ and $\lambda=0.5$.]
    {
        \includegraphics[scale=0.3]{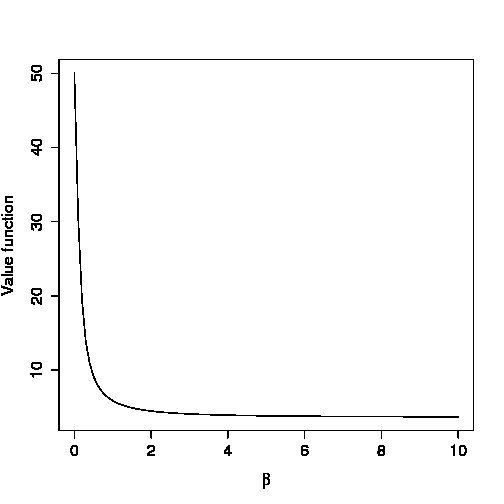}
        \label{betaBM}
    }
    \\
    \subfigure[Change in $V(5,2)$ as $\mu$ varies and $\lambda=0.5$, $\sigma=0.5$ and $\beta=1$.]
    {
        \includegraphics[scale=0.3]{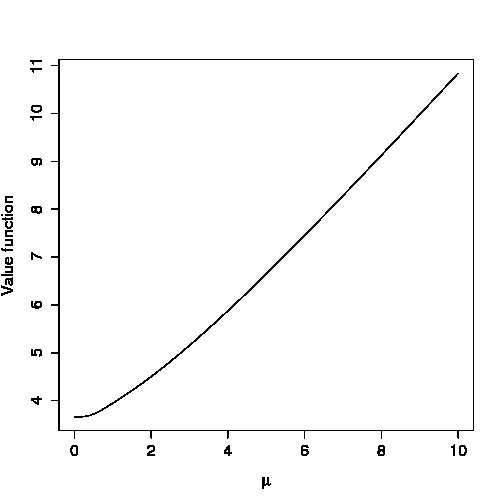}
        \label{muBM}
    }
    \subfigure[Change in $V(5,2)$ as $\sigma$ varies and $\mu=4$, $\lambda=0.5$ and $\beta=1$.]
    {
        \includegraphics[scale=0.3]{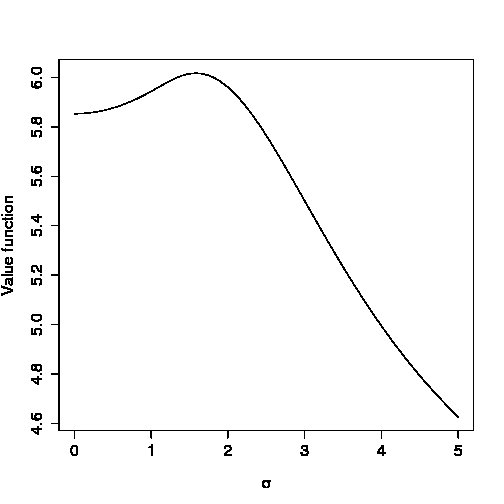}
        \label{sigmaBM}
    }
    \caption{Change in the parameters of the model BM.}
    \label{Sensi}
\end{figure}

For $\mu\leq0$ it is not optimal to wait at all, so $V=W$, but as $\mu$ increases clearly the value function increases in an almost linear fashion.

The effect of $\sigma$ in the value function is probably the most interesting one. In figure \ref{sigmaBM} we see that it is beneficial for the investor to have some variance in the asset but not too much. An explanation for this is that when the variance increases it is more likely for the price process to enter the trading region. On the other hand, if the variance is too big, the process can hit 0 too fast. Clearly the variance of the revenue increases with $\sigma$, thus as part of future research it would be interesting to consider the risk aversion of the investor.

\begin{figure}
  \centering
  \subfigure[Value function in the OU case with parameters $\lambda=0.5$, $\alpha=4$, $\sigma=0.5$, $m=5$ and $\beta=1$.]   
  {
	\includegraphics[scale=0.4]{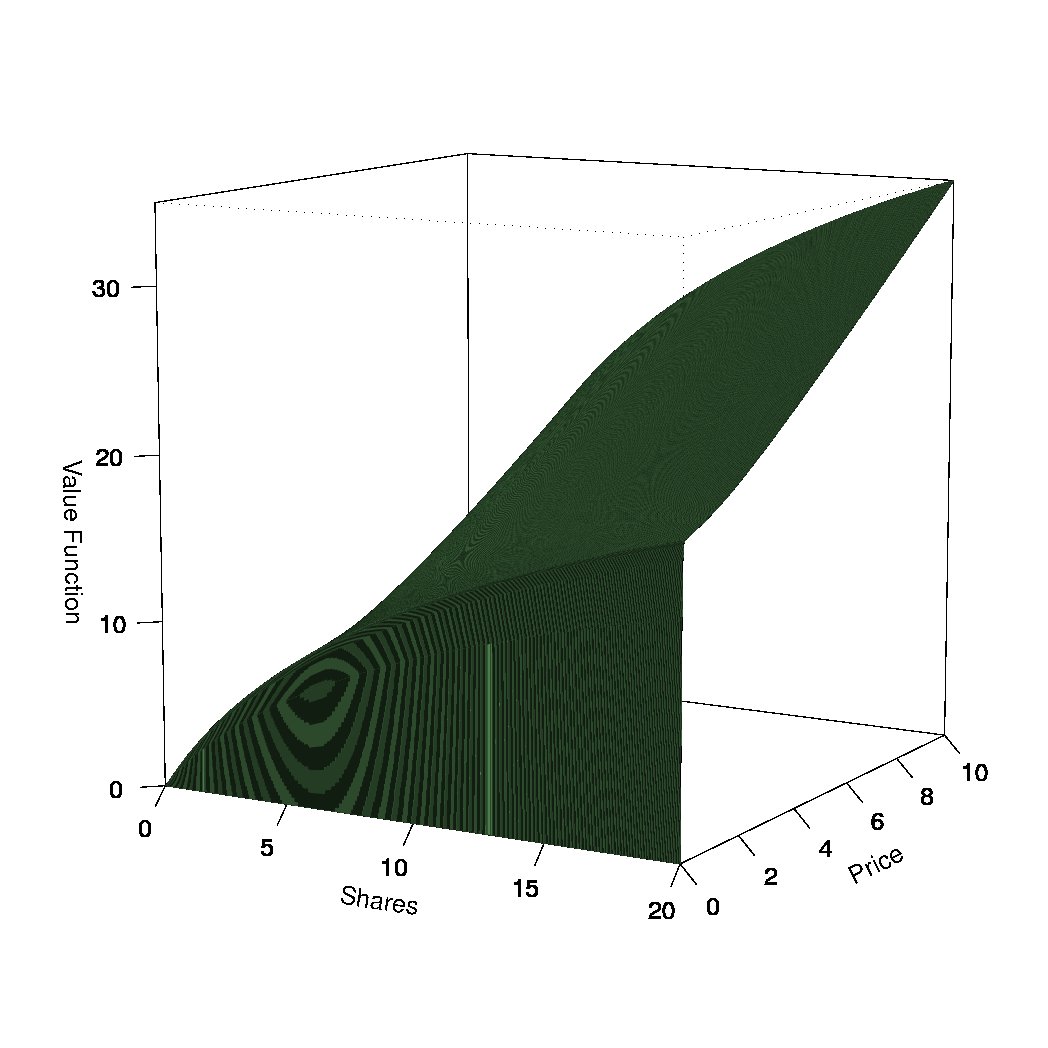}
        \label{valueOU}
  }
  \subfigure[Continuation-trade region in the OU case. The solid line shows the contour with parameters $\lambda=0.5$, $\alpha=4$, $\sigma=0.5$, $m=5$ and $\beta=1$. In the other lines only the indicated parameter has been changed.]
  {
  	\includegraphics[scale=0.4]{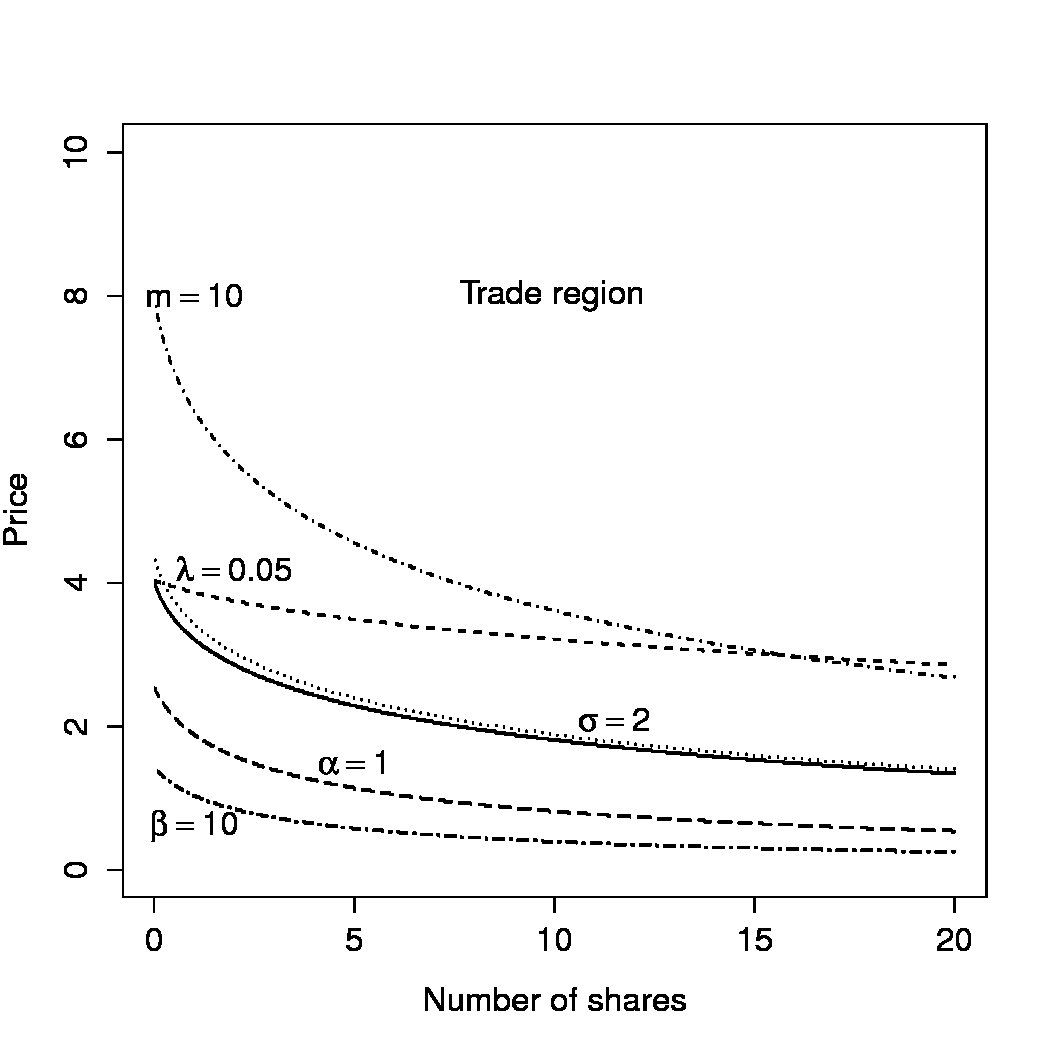}
	\label{tradeOU}	
  }
  \caption{Value function and continuation-trade region in the mean-reverting case.}
  \label{OU}
\end{figure}

The second example is when the price follows the Ornstein–Uhlenbeck process. Then the price process becomes
$$dP_t=\alpha(m-P_{t})dt+\sigma dB_t-\lambda P_{t}d\xi_t,$$
with $\sigma,\alpha>0$. As in the case of arithmetic Brownian motion, the boundary conditions are given by \eqref{boundaryinfab}, since 0 is an absorbing boundary for this process. Figure \ref{OU} shows the value function and the continuation-trade region. Again, this case does not fit within Theorem \ref{conecspecial}, so the strategy is similar to the BM case. Also, the sensitivity of the function to the parameters is similar to the previous case. The only parameter that is exclusive to the mean-reverting process is the resilience factor $\alpha$. As we increase $\alpha$ the value function increases (Figure \ref{alpha}) and the continuation region grows (Figure \ref{tradeOU}).

\begin{figure}
    \centering
    \subfigure[Change in $V(5,2)$ as $\lambda$ varies and $m=5$, $\sigma=0.5$, $\alpha=4$ and $\beta=1$.]
    {
        \includegraphics[scale=0.3]{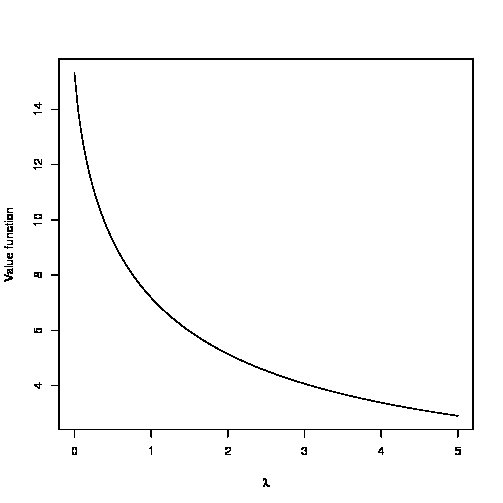}
        \label{lambdaOU}
    }
    \\
    \subfigure[Change in $V(5,2)$ as $\beta$ varies and $m=5$, $\sigma=0.5$, $\alpha=4$ and $\lambda=0.5$.]
    {
        \includegraphics[scale=0.3]{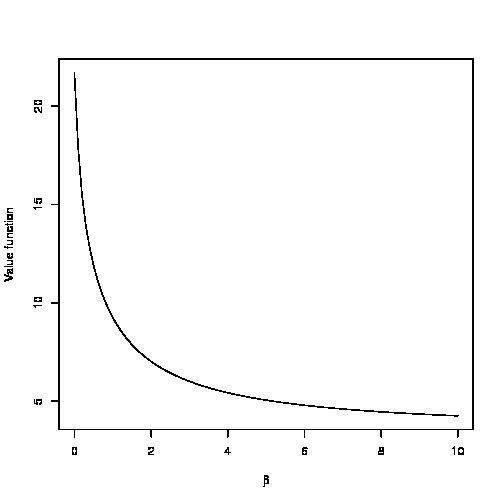}
        \label{betaOU}
    }
    \subfigure[Change in $V(5,2)$ as $m$ varies and $\alpha=4$, $\sigma=0.5$, $\beta=1$ and $\lambda=0.5$.]
    {
        \includegraphics[scale=0.3]{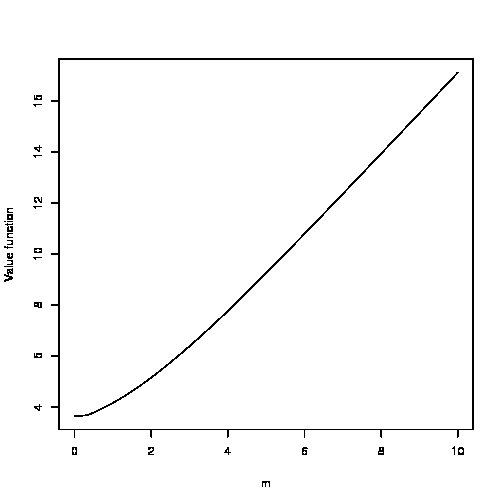}
        \label{mean}
    }
    \\
    \subfigure[Change in $V(5,2)$ as $\alpha$ varies and $\lambda=0.5$, $\sigma=0.5$, $m=5$ and $\beta=1$.]
    {
        \includegraphics[scale=0.3]{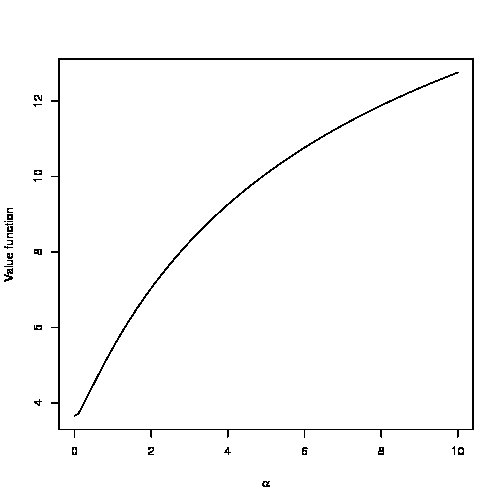}
        \label{alpha}
    }
    \subfigure[Change in $V(5,2)$ as $\sigma$ varies and $\alpha=4$, $\lambda=0.5$, $m=5$ and $\beta=1$.]
    {
        \includegraphics[scale=0.3]{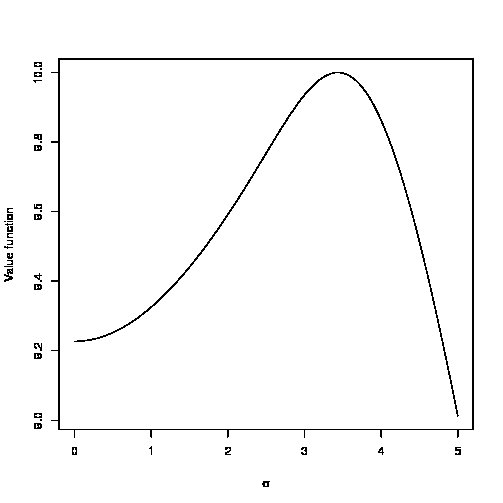}
        \label{sigmaOU}
    }
    \caption{Change in the parameters of the model OU.}
    \label{SensiOU}
\end{figure}

\section{Conclusions}

The main goal of this work was to characterize the value function of the optimal execution strategy in the presence of price impact and fixed transaction cost over an infinite horizon. We formulated the problem using two different stochastic control settings. In the impulse control formulation we showed that the value function is the unique continuous viscosity solution of the Hamilton-Jacobi-Bellman equation associated to the problem whenever the transaction cost is strictly positive. The second formulation ruled out any transaction cost and admitted continuous singular controls only. In this case we also proved continuity and uniqueness of the value function under the viscosity framework. The next step, part of future research, would be to find the regularity of the value function. Numerical results provided in this paper, at least for the second formulation, suggest that the function is more than just continuous and that its regularity is related with the regularity of the function $U$ defined in Section \ref{model}. Although any impulse control is a singular control, in general the expected revenue obtained when applying the same impulse control in both formulation is different. However, the value function may be the same. In fact, we were able to show that this is the case for a special type of price impact and provide the explicit solution. The question if this is true in general is still unanswered. This is particularly challenging since the subsolution property for the HJB equation \eqref{hjbinf}, when there is no transaction cost, has no information at all. Try to find answers is part of future work. Another important conclusion is that the HJB equation for the regular control formulation, \eqref{hjbreg} has not enough information to characterize the solution.  From an economic viewpoint, it would be important to study the effect of the price impact in hedging strategies and how they are different to the strategies obtained in classical models, e.g. Delta-hedging in Black-Scholes setting. Include utility functions to account for risk aversion is another important extension of this work. Finally, the finite time horizon natural extension is currently in preparation.
\appendix
\section{Proof of $(\mathcal{M}V)^*\leq\mathcal{M}V^*$}\label{lemita}

Let $\varphi$ be a locally bounded function on $\bar{\mathcal{O}}$. Let $(y_n)$ be a sequence in $\mathcal{O}$ such that $(y_n)\rightarrow y_0$ and
$$\lim_{n\rightarrow\infty}\mathcal{M}\varphi(y_n)=(\mathcal{M}\varphi)^*(y_0).$$
Since $\varphi^*$ is usc and $\Gamma$ is continuous, for each $n\geq1$ there exists $0\leq\zeta_n\leq x_n$ such that
$$\mathcal{M}\varphi^*(y_n)=\varphi^*(\Gamma(y_n,\zeta_n))+\zeta_n\alpha(\zeta_n,p_n)-k.$$
The sequence $(\zeta_n)$ is bounded (since $x_n\rightarrow x_0$) and therefore converges along a subsequence to $\zeta\in[0,x_0]$. Hence
\begin{align*}
(\mathcal{M}\varphi)^*(y_0)&=\lim_{n\rightarrow\infty}\mathcal{M}\varphi(y_n)\\
&\leq\limsup_{n\rightarrow\infty}\mathcal{M}\varphi^*(y_n)\\
&=\limsup_{n\rightarrow\infty}\varphi^*(\Gamma(y_n,\zeta_n))+\zeta_n\alpha(\zeta_n,p_n)-k\\
&\leq\varphi^*(\Gamma(y_0,\zeta))+\zeta\alpha(\zeta,p_0)-k\\
&\leq\mathcal{M}\varphi^*(y_0).
\end{align*}

\bibliography{mybib}

\begin{thebibliography}{28}
\providecommand{\natexlab}[1]{#1}
\providecommand{\url}[1]{\texttt{#1}}
\expandafter\ifx\csname urlstyle\endcsname\relax
  \providecommand{\doi}[1]{doi: #1}\else
  \providecommand{\doi}{doi: \begingroup \urlstyle{rm}\Url}\fi

\bibitem[Almgren(2003)]{Almgren2}
R.~Almgren.
\newblock Optimal execution with nonlinear impact functions and
  trading-enhanced risk.
\newblock \emph{Applied Mathematical Finance}, 10\penalty0 (1):\penalty0 1--18,
  January 2003.

\bibitem[Almgren and Chriss(2000)]{Almgren}
R.~Almgren and N.~Chriss.
\newblock Optimal execution of portfolio transactions.
\newblock \emph{Journal of Risk}, 3:\penalty0 5--39, 2000.

\bibitem[Bank and Baum(2004)]{Bank}
P.~Bank and D.~Baum.
\newblock {Hedging and portfolio optimization in financial markets with a large
  trader.}
\newblock \emph{Math. Finance}, 14\penalty0 (1):\penalty0 1--18, 2004.

\bibitem[Bertsimas and Lo(1998)]{Bertsimas}
D.~Bertsimas and A.~W. Lo.
\newblock Optimal control of execution costs.
\newblock \emph{Journal of Financial Markets}, 1\penalty0 (1):\penalty0 1--50,
  April 1998.

\bibitem[\c{C}etin et~al.(2004)\c{C}etin, Jarrow, and Protter]{CetinLiq}
U.~\c{C}etin, R.~Jarrow, and P.~Protter.
\newblock {Liquidity risk and arbitrage pricing theory.}
\newblock \emph{Finance Stoch.}, 8\penalty0 (3):\penalty0 311--341, 2004.

\bibitem[Chan and Lakonishok(1995)]{Chan}
L.K.C. Chan and J.~Lakonishok.
\newblock The behavior of stock prices around institutional trades.
\newblock \emph{Journal of Finance}, 50\penalty0 (4), September 1995.

\bibitem[Crandall et~al.(1992)Crandall, Ishii, and Lions]{crisli}
M.~G. Crandall, H.~Ishii, and P.~L. Lions.
\newblock User's guide to viscosity solutions of second order partial
  differential equations.
\newblock \emph{Bulletin of the American Mathematical Society}, 27\penalty0
  (12):\penalty0 1--67, July 1992.

\bibitem[Cvitani{\'c} and Ma(1996)]{Cvitanic}
J.~Cvitani{\'c} and J.~Ma.
\newblock {Hedging options for a large investor and forward-backward SDE's.}
\newblock \emph{Ann. Appl. Probab.}, 6\penalty0 (2):\penalty0 370--398, 1996.

\bibitem[Davis and Norman(1990)]{Davis}
M.H.A. Davis and A.R. Norman.
\newblock {Portfolio selection with transaction costs.}
\newblock \emph{Math. Oper. Res.}, 15\penalty0 (4):\penalty0 676--713, 1990.

\bibitem[Dayanik and Karatzas(2003)]{dayanik}
S.~Dayanik and I.~Karatzas.
\newblock On the optimal stopping problem for one-dimensional diffusions.
\newblock \emph{Stochastic Processes and their Applications}, 107\penalty0
  (2):\penalty0 173--212, October 2003.

\bibitem[Fleming and Soner(2006)]{Fleming}
W.H. Fleming and H.M. Soner.
\newblock \emph{Controlled Markov Processes and Viscosity Solutions}.
\newblock Springer, second edition, 2006.

\bibitem[Frey(1998)]{Frey}
R.~Frey.
\newblock Perfect option hedging for a large trader.
\newblock \emph{Finance and Stochastics}, 2\penalty0 (2):\penalty0 115--141,
  1998.

\bibitem[He and Mamaysky(2005)]{He}
H.~He and H.~Mamaysky.
\newblock Dynamic trading policies with price impact.
\newblock \emph{Journal of Economic Dynamics and Control}, 29\penalty0
  (5):\penalty0 891--930, 2005.

\bibitem[Holthausen et~al.(1990)Holthausen, Leftwich, and Mayers]{Holthausen}
R.W. Holthausen, R.~Leftwich, and D.~Mayers.
\newblock Large-block transactions, the speed of response, and temporary and
  permanent stock-price effects.
\newblock \emph{Journal of Financial Economics}, 26\penalty0 (1):\penalty0
  71--95, July 1990.

\bibitem[Ishii and Lions(1990)]{ishii1}
H.~Ishii and P.~L. Lions.
\newblock Viscosity solutions of fully nonlinear second-order elliptic partial
  differential equations.
\newblock \emph{Journal of Differential Equations}, 83\penalty0 (1):\penalty0
  26 -- 78, 1990.

\bibitem[Ishii(1993)]{ishii2}
K.~Ishii.
\newblock Viscosity solutions of nonlinear second order elliptic pdes
  associated with impulse control problems.
\newblock \emph{Funkcialaj Ekvacioj}, 36:\penalty0 123--141, 1993.

\bibitem[Ishikawa(2004)]{Ishikawa}
Y.~Ishikawa.
\newblock Optimal control problem associated with jump processes.
\newblock \emph{Applied Mathematics and Optimization}, 50\penalty0
  (1):\penalty0 21--65, 2004.

\bibitem[Korn(1998)]{Korn}
R.~Korn.
\newblock Portfolio optimisation with strictly positive transaction costs and
  impulse control.
\newblock \emph{Finance and Stochastics}, 2\penalty0 (2):\penalty0 85--114,
  1998.

\bibitem[Kushner and Dupuis(1992)]{Kushner}
H.~Kushner and P.~Dupuis.
\newblock \emph{Numerical methods for stochastic control problems in continuous
  time}.
\newblock Springer-Verlag, New York, 1992.

\bibitem[Ly~Vath et~al.(2007)Ly~Vath, Mnif, and Pham]{Ly}
V.~Ly~Vath, M.~Mnif, and H.~Pham.
\newblock {A model of optimal portfolio selection under liquidity risk and
  price impact.}
\newblock \emph{Finance Stoch.}, 11\penalty0 (1):\penalty0 51--90, 2007.

\bibitem[Ma and Yong(1999)]{Ma}
Jin Ma and Jiongmin Yong.
\newblock Dynamic programming for multidimensional stochastic control problems.
\newblock \emph{Acta Mathematica Sinica}, 15:\penalty0 485--506, 1999.
\newblock ISSN 1439-8516.

\bibitem[{\O}ksendal and Reikvam(1998)]{oksenreik}
B.~{\O}ksendal and K.~Reikvam.
\newblock Viscosity solutions of optimal stopping problems.
\newblock \emph{Stochastics An International Journal of Probability and
  Stochastic Processes: formerly Stochastics and Stochastics Reports},
  62\penalty0 (3):\penalty0 285-- 301, 1998.

\bibitem[{\O}ksendal and Sulem(2002)]{oksulem}
B.~{\O}ksendal and A.~Sulem.
\newblock Optimal consumption and portfolio with both fixed and proportional
  transaction costs.
\newblock \emph{SIAM Journal on Control and Optimization}, 40\penalty0
  (6):\penalty0 1765--1790, 2002.

\bibitem[{\O}ksendal and Sulem(2005)]{oeksendal05acjd}
B.~{\O}ksendal and A.~Sulem.
\newblock \emph{Applied Stochastic Control of Jump Diffusions}.
\newblock Springer, 2005.

\bibitem[Protter(2004)]{ProtterBook}
P.~Protter.
\newblock \emph{Stochastic Integration and Differential Equations}.
\newblock Springer-Verlag, Berlin, second edition, 2004.

\bibitem[Schied and Sch{\"o}neborn(2009)]{Schied2}
A.~Schied and T.~Sch{\"o}neborn.
\newblock {Risk aversion and the dynamics of optimal liquidation strategies in
  illiquid markets.}
\newblock \emph{Finance Stoch.}, 13\penalty0 (2):\penalty0 181--204, 2009.

\bibitem[Schied et~al.(2010)Schied, Sch{\"o}neborn, and Tehranchi]{Schied3}
A.~Schied, T.~Sch{\"o}neborn, and M.~Tehranchi.
\newblock {Optimal Basket Liquidation for CARA Investors is Deterministic}.
\newblock \emph{Applied Mathematical Finance}, 17:\penalty0 471--489, 2010.

\bibitem[Subramanian and Jarrow(2001)]{ajay}
A.~Subramanian and R.~Jarrow.
\newblock The liquidity discount.
\newblock \emph{Mathematical Finance}, 11\penalty0 (4):\penalty0 447--474,
  2001.

\end{thebibliography}
\bibliographystyle{plainnat}

\end{document}